\documentclass[num-refs]{nbdt-article}
\usepackage{hyperref}       % hyperlinks

\usepackage{siunitx}

\usepackage[utf8]{inputenc} % allow utf-8 input
\usepackage[T1]{fontenc}    % use 8-bit T1 fonts
\usepackage{url}            % simple URL typesetting
\usepackage{booktabs}       % professional-quality tables
\usepackage{amsfonts}       % blackboard math symbols
\usepackage{nicefrac}       % compact symbols for 1/2, etc.
\usepackage{microtype}      % microtypography
\usepackage{amsmath}
\newcommand{\diff}{\mathop{}\!\mathrm{d}}
\usepackage{algorithm}
\usepackage{algpseudocode}
\usepackage{graphicx}
\usepackage{wrapfig}
\usepackage{float}
\usepackage{caption}
\usepackage{subcaption}
\usepackage{mathtools}
\usepackage{xcolor}
\newtheorem{prop}{Proposition}
\newcommand{\defeq}{\vcentcolon=}
\usepackage{soul}

\usepackage{amsfonts}
\newcommand{\R}{\mathbb{R}}

\newcommand{\EcommentDone}[1]{}

\papertype{Original Article}
\paperfield{Journal Section}

\title{%Approaches for Direct Training of Biophysical Neural Networks: 
Evolutionary algorithms as an alternative to backpropagation for supervised training of Biophysical Neural Networks and Neural ODEs}

\abbrevs{BBPES (Biophysical BackPropagation and Evolutionary Strategies)}

\author[1,3]{James Hazelden}
\author[1,3]{Yuhan Helena Liu}
\author[1,2,3*]{Eli Shlizerman}
\author[1,2,3,4*]{Eric Shea-Brown}

\affil[1]{Applied Mathematics and Computational Neuroscience Center, University of Washington, Seattle, WA, 98195, United States}
\affil[2]{Electrical \& Computer Engineering, University of Washington, Seattle, WA, 98195, United States}
\affil[3]{Computational Neuroscience Center, University of Washington, Seattle, WA, 98195, United States}
\affil[4]{Physiology \& Biophysics, University of Washington, Seattle, WA, 98195, United States}
\affil[*]{\textbf{These authors share senior authorship}}

\corraddress{James Hazelden, Applied Mathematics, University of Washington, Seattle, WA, 98195, United States}
\corremail{jhazelde@uw.edu}

\fundinginfo{}

\runningauthor{Hazelden et al.}

\begin{document}

\maketitle

\begin{abstract}
Training networks consisting of biophysically accurate neuron models could allow for new insights into how brain circuits can organize and solve tasks. We begin by analyzing the extent to which the central algorithm for neural network learning -- stochastic gradient descent through backpropagation (BP) -- can be used to train such networks. 
We find that properties of biophysically based neural network models needed for accurate modelling such as stiffness, high nonlinearity and long evaluation timeframes relative to spike times makes BP unstable and divergent in a variety of cases.
To address these instabilities and inspired by recent work, we investigate the use of ``gradient-estimating'' evolutionary algorithms (EAs) for training biophysically based neural networks. 
We find that EAs have several advantages making them desirable over direct BP, including being forward-pass only, robust to noisy and rigid losses, allowing for discrete loss formulations, and potentially facilitating a more global exploration of parameters. 
We apply our method to train a recurrent network of Morris--Lecar neuron models on a stimulus integration and working memory task, and show how it can succeed in cases where direct BP is inapplicable. 
To expand on the viability of EAs in general, we apply them to a general neural ODE problem and a stiff neural ODE benchmark and find again that EAs can out-perform direct BP here, especially for the over-parameterized regime.
\textbf{Our findings suggest that biophysical neurons could provide useful benchmarks for testing the limits of BP-adjacent methods, and demonstrate the viability of EAs for training networks with complex components. }

% Please include a maximum of seven keywords
\end{abstract}
\section{Introduction}

\EcommentDone{should we add ``and neural ODEs'' to the end of title?  think of ways to tie into main flow better ... bit of stretch but could say is step toward modeling more general bio processes / nets?}

The brain is able to compute an astonishing variety of tasks through the networked interaction of many spiking neurons.  Biophysical models for these neurons can demonstrate diverse ranges of behavior, described through, for example, a variety of bifurcation types, \textit{resonance} and \textit{integration} behavior, and bursting or tonic spiking \citep{izhikevich2004which, gerstner2014neuronal}. 
These mechanisms can be used in diverse ways to facilitate communication in the brain. 
For example, resonator neurons can ``multiplex'' input signals, tonic spiking can be utilized to communicate through rate-coding, and the cortical neurons exhibit spiking to bursting transitions during sleep \citep{Izhikevich2007, steriade2001sleep}.  
Understanding how these biophysical models can be coupled together to solve information processing tasks, even primitive, is a fascinating challenge that pervades computational and experimental neuroscience. 

In this work, we ask to what extent iterative supervised learning algorithms can train networks of biophysical neuron models to solve tasks.  
%If such learning is possible, it could enable much insight into the questions posed above, allowing researchers to better navigate and optimize the wide range of behaviors these networks can exhibit. 
Stochastic gradient descent through backpropagation (BP) is the primary tool that has enabled the massive success of supervised learning for artificial neural networks \cite{rumelhart1986learning}.
This approach computes the required gradients by reversing (backpropagating) the activity that flowed through the network during a forward pass. This typically relies on the neurons implementing differentiable operations with well-defined derivatives, and assembles an  explicit composition (the derivative of operations) through the network. A common trend has involved extending backpropgation to networks with increasingly diverse components, leading to fruitful new directions. Some examples include spiking neural networks \citep{zenke2018surrogate}, neural ODEs \citep{chen2018neural} and binary neural networks \citep{rastegari2016xnornet}.

Biophysical models, being multi-dimensional nonlinear dynamical systems with widely disparate timescales, exhibit numerical stiffness, instability and rapid, almost discontinuous, dynamics. The major source of these features is the underlying ``spikes,'' in which voltage and sometimes allied state variables change extremely quickly.
%wherein the state variables can change extremely quickly depending on the physiological parameters. 
Such dynamics pose challenges for BP, as it accumulates gradients over time and generally requires well-behaved, differentiable states at every timestep that do not blow up. 
% \EcommentDone{Somewhere around this para, I think would be good to remind readers of the challenges that BP poses, in terms of its demands ( sensitivity based on derivatives, accumulated across network and espeically time in the case of ODE models and discrete nets) ... this will set up what follows and help explain why these stiffness issues arise in BP in particular}

These numerical challenges 
% \EcommentDone{different word than complexity I think -- attendant numerical challenges or something?} 
may be amplified when multiple of these models interact in a network. 
Further accentuating this problem, many settings relevant to neurobiology carry complex or discrete loss functions, and are evaluated over long timeframes (seconds or even minutes) in which very many spikes occur; furthermore, models could include sources of stochasticity.
For such situations, it has recently been observed that BP can become less reliable due to dependence on the precise timing of backpropagated derivatives, sensitivity to instabilities that these models exhibit, and gradients that are not well defined in the case of non-differentiable losses or neuron models \citep{yingbo2021comparison, kim2021stiff}. In this work, we begin by asking \textbf{whether BP is applicable for networks of neurons that incorporate biophysically based neuron models,} i.e. neurons expressed by 
Hodgkin--Huxley differential equations or their simplified variants, such as Morris--Lecar, FitzHugh Nagumo equations, etc \citep{hodgkin1952quantative, morris1981voltage, fitzhugh1955mathematical}. 
To this end, our work first focuses on investigating the applicability of BP and its variants (BPTT and the neural ODEs approach, explained in Methods section below) in this setting \citep{chen2018neural}.
% In this work, we demonstrate these complex scenarios by considering the training of networks of biophysical neuron models.

To investigate the applicability and limitations of BP, we introduce two new simple but useful benchmarks: 
\begin{enumerate}
    \item A single Morris--Lecar \citep{morris1981voltage} neuron in which the external current is a learnable parameter and the goal is to minimize spiking of the neuron. Section \ref{casestudy} outlines the motivation and details of this choice.
    \item A coupled Morris--Lecar network applied to a common working memory and stimulus integration task, as in Section \ref{integrator_task}.
\end{enumerate}
\noindent{We illustrate certain settings in which BP succeeds in training networks to solve these tasks, and other regimes in which it is unreliable and extremely unstable.  We further note how it can be prohibitively memory intensive for realistically long timeframes.}
%, using standard BP variants such as direct Back Propagation Through Time (BPTT) or the neural ODEs approach \citep{chen2018neural, werbos1988learning}. 

To address the shortcomings of BP for these situations, we follow the direction of much recent exciting research on the surprising viability of evolutionary algorithms (EAs) as an alternative to BP-based approaches \citep{such2017deep, bai2023evolutionary, salimans2017evolution}. 
In particular, we describe how evolutionary strategies (ES), used previously by, e.g., Salimans, et al. \citep{salimans2017evolution}, can be adapted to the problem of supervised training for biophysical neural networks (BNNs) and neural ODEs in general. 
% Specifically we introduce a novel approach, Smooth Observational Neural Gradients (ES), for gradient estimation for supervised training of these networks.
Inspired by the stochasticity and distributed processing observed in biological systems, 
ES leverage Monte Carlo gradient estimation and "smoothing" of the loss landscape. 
In particular, ES first introduce a new ``smoothed'' loss function that is computed through the expected value of the original loss under variations in the parameters of the network according to a chosen distribution. 
This modification naturally makes ES more robust for stiff problems since it considers trends in a neighborhood for gradient computation, instead of at a singular point. 
As we describe in more depth below (Section \ref{derivation}), ES permits a gradient estimate based on the \textit{log-trick}, allowing for efficient gradient computation using minimal samples and scaling well for large numbers of parameters \citep{williams1992simple, salimans2017evolution}. 
% Next, we find that the new smooth loss permits a Monte-Carlo estimate by applying the so-called \textit{log-trick} proposed in reinforcement learning \citep{williams1992simple}. Monte-Carlo scales better with the number of network parameters than direct numerical methods such as finite differences. 
We find that evolutionary algorithms are well suited for training neural networks with complex components:
% \EcommentDone{rephrase ``this sort'' (not clear what refers to)} 
they naturally filter noise, are robust to stiff biological problems and are forward-pass only, alleviating some of the challenges that BP faces related to exploding gradients and non-differentiability.

Here, we evaluate the applicability of ES (specifically the evolutionary strategy described) for supervised learning in the benchmark settings (1) and (2) above, and contrast it with BP.  
% \EcommentDone{somewhere in here, should we have citations to weight perturbation?  guessing this will come to mind quickly of many comp neuroscientists.  is it Fiete, MIT, who studied this prominently?} 
%to the biophysical benchmarks (1) and (2) above and compare their performance in different situations. 
In some cases, BP works best and converges efficiently, while in others BP diverges or is constrained by memory problems; ES often continue to perform in these cases. %Specifically, we find that the neural ODEs approach is not applicable to these models since simulating them backwards in time will tend to diverge extremely quickly. In contrast, using the BPTT approach for BP which stores all forward pass evaluations seems to perform well in certain cases but is memory constrained for longer timeframes. However, we find that BPTT can produce explosive gradients when the biophysical model parameters are made more fast (faster spike transitions and firing rates).
We believe that our results on these two benchmarks by provide theoretical insights into when BP can be used and the benefits of ES as an alternative to BP. 

% We find that gradient problems \EcommentDone{rephrase} are more prominent when long timeframes or stiffer neuronal parameters are used.  

Finally, we analyze the applicability of ES in the broader context of neural ODEs (the ``continuous time'' analogue of recurrent neural networks, where individual components are given by solving systems of ODEs). \EcommentDone{define what we mean by this}. 
We first show that ES can be effective on a simple cubic oscillator neural ODE and, compared to BP, can achieve faster convergence to a lower loss in overparameterized settings.
%This conclusion is well-supported in the current literature which suggests EAs scale better than BP for overparameterized neural networks. 
We further apply to the standard ``ROBER'' stiff ODE problem, in which BP has been shown to be unstable \citep{robertson1967, kim2021stiff}. 
We find that ES can train the neural ODE in this situation without having to adjust the loss to account for the varied sources of stiffness. 

% \textcolor{blue}{Please see section \ref{sec:relatedwork} for an expanded discussion of related literature, including {xx yy zz ... add a couple of highlights (few words each) for comp neuro minded readers}. }

In sum, our \textbf{main contributions} below are as follows:
\begin{itemize}
    \item The introduction of two new simple testbeds for evaluating the applicability of BP and other supervised learning approaches for training biophysical neuronal models individually and when coupled together into networks. 
    \item We apply BP to these two testbeds and determine where it can be unstable, giving corrupted gradients, and what causes this instability (long timeframes, spike time resolution in networks, and exploding gradients). Our findings are encapsulated in Proposition \ref{prop1} and \ref{prop2}.
    \item We describe and apply evolutionary strategies (ES) in these settings and demonstrate that they are a more reliable alternative to BP with desirable properties (noise filtering, no dependence on forward pass or differentiability and robustness to rapid changes in loss). 
    \item We also apply ES to more generic neural ODE problems: a cubic oscillator, in which we find that ES can outperform BP when the network is over-parameterized, and a standard stiff ODE problem, where ES appear to be more stable than BP. These results indicate a need for further investigation into using ES for neural ODEs. 
    % \item Analysis of the applicability of BP on a single Morris--Lecar neuron, exhibiting situations causing BP to fail and the sources of stiffness in this problem (fast or long timeframes). 
    % \item Investigation of BP for training a network of Morris--Lecar neurons on a
    % \item Introduced a smoothed loss function, which helps to overcome non-differentiable and irregular loss landscapes and improves the stability of the training process (Section \ref{derivation}). 
    % \item Developed a novel learning rule, ES, that utilizes the smoothed loss to estimate gradients, simultaneously bypassing (1) the need for a backwards pass and differentiable output of BP, as well as (2) the instability issues in training networks with biologically realistic neuron models (Section \ref{derivation}).  
    % \item Provided theoretical support for the effectiveness of the learning rule (Proposition~\ref{thm:decE}).
    % \item Demonstrated the effectiveness of ES by successfully training a network of biophysical neurons to solve a working memory task (Section \ref{integrator_task}) and showcasing improved stability and performance on a stiff neural ODEs (Section \ref{stiff}).
    % \item Demonstrated the general applicability of our sampling approach to more common neural ODE problems and that ES scales far better than naive sampling approaches such as FD (Section \ref{neuralodes}).
    % \item \textcolor{blue}{Should we mention anything regarding the scalability in terms of number of samples? Depends on how strong the results are and how comfortable you feel about emphasizing them. }
\end{itemize}

\subsection{Related Work}

\label{sec:relatedwork}

\textbf{Biophysical Neural Computation} A central pillar of computational neuroscience is understanding the dynamics and computation that arise in circuits with biophysically accurate neuronal components and synapses.
% Related to integrator tasks, Wang proposed an architecture that solves the continuous integration problem (so-called DOTs task) with networks of leaky integrate-and-fire components in 2002. Since then, integrator tasks have been used as a means of understanding neural computation and formation of working memory. 
Within it, direct supervised learning approaches for training of biophysical neuronal networks have largely been focused on reduced neuron models such as integrate-and-fire, resonate-and-fire, and the Izhikevich neuron. Prior work demonstrates that these can be trained with the ``surrogate gradient'' approach together with direct BP \citep{zenke2018surrogate,alkhamissi2021deep,bellec2020solution}. However, efficiently training Hodgkin--Huxley and Morris--Lecar models in networks remains unresolved. The Hodgkin--Huxley system, for example, requires three orders of magnitude more operations per ODE step than integrate-and-fire \cite{izhikevich2004which}.  Moreover, it has four state variables that undergo complex and variable dynamics: its various parameter settings can induce resonance, bifurcation, and bistability.  This likely contributes to the entire system being more difficult to train. On the other hand, this dynamical complexity is of high interest~\cite{winston2023heterogeneity,perez2021neural,renart2003robust,billeh2020systematic,litwin2016inhibitory,liu2020excitation,zeldenrust2021efficient,whittington2023disentanglement,salaj2021spike,stockl2021structure}, as it could enable network computation through mechanisms such as multiplexing and bursting that are not possible with simplified integrate-and-fire models \citep{Izhikevich2007}.

\textbf{Evolutionary Algorithms for Deep Learning} There has been a newfound interest in adapting and using evolutionary algorithms, specifically ``evolutionary strategies'' (ES) in deep and reinforcement learning (RL) contexts. Salimans et al. found that evolutionary strategies in which the gradient in descent is estimated approximately by local sampling can outperform standard RL in terms of convergence versus runtime \citep{salimans2017evolution}. Evolutionary approaches have also been merged with RL to create hybrid learning rules \citep{bai2023evolutionary}. It has been noted, in particular, that ES can be very useful for learning problems where the ``loss landscape'' being traversed is very noisy or discrete \citep{li2023noise}. ES can also effectively avoid local minima, since they are unrestricted in their exploration rule. The evolutionary strategy used in this work is a specific type of evolutionary algorithm. It is more broadly referred to as a ``score function estimator'' in machine learning theory. Score function estimators estimators are well studied as an approach for approximating gradients of expected values \citep{mohamed2020monte}. Less work has investigated applications of ES to direct supervised learning problems, such as those encountered with artificial neural networks, and to the best of our knowledge there has been minimal work investigating the contexts we consider here: biophysical neural networks and, more generally, neural ODEs. In the latter case, there is work using estimators such as finite difference to compute gradients, but as far as we are aware, evolutionary strategies have not been well studied for neural ODEs. 
% Our work focuses on applying EAs in the context of biophysical neural networks and general neural ODEs. To the best of our knowledge, these applications of EAs have not been well investigated in the literature. 

% \textbf{Biologically Plausible Learning Rules:} We can think of EAs as an 

\textbf{Backpropagation Alternatives} The context in which we apply our method naturally prompts questions about the biological plausibility of ES methods. In particular, if the ES is a viable alternative to BP for training biophysical neural networks, it may not be out of the question that they could be implemented in real brains.  In this work, we consider the use of the evolutionary strategy (ES) of Salimans et al. as a gradient estimator and backpropagation alternative. ES is based on random, small, perturbations to the synaptic connections and quantification of the change in the final loss. Certain strategies only require a single such evaluation to update the network connection weights, and ES can fundamentally function in the same way (although the variance could be high; see theory below in Section \ref{derivation}). In practice, ES seem to require many samples to get good gradient estimates that are low in variance currently, a feature counting against the biological plausibility when contrasted with simpler evolutionary approaches such as weight perturbation.

The theory and implementation of biologically plausible mechanisms that could facilitate BP in the brain is an active and growing area of research.  Overall, the implementation of BP %relies on the efficient computation of gradients for adjusting network parameters, but its implementation 
in the brain remains a topic of debate~\cite{diederich1987learning,lillicrap2020backpropagation,richards2019deep,scellier2017equilibrium,hinton2022forward,laborieux2022holomorphic,meulemans2022least}. Investigating biologically plausible learning rules can lead to a better understanding of how the brain processes information and learns, and could potentially inform the development of new artificial neural network training methods. Such investigations have already yielded interesting results in AI-inspired modeling in neuroscience~\cite{zador2023catalyzing,richards2019deep,sorscher2023unified,hennequin2018dynamical,vyas2020computation,perich2021inferring,yang2019task} and neuromorphic computing~\cite{schuman2022opportunities,strubell2020energy,covi2021adaptive,cramer2022surrogate,boahen2022dendrocentric,roy2019towards}. 
% Beyond questions of implementation, algorithms that function like the brain can also have computational benefits~\cite{lake2017building}.
Many of these approaches attempt to approximate BP through neural mechanisms, e.g., by truncating and approximating intermediate derivatives in BP~\cite{murray2019local,bellec2020solution,liu2021cell,liu2022biologically,marschall2020unified,roelfsema2018control,payeur2020burst,sacramento2018dendritic}.  Other work uses alternative learning rules (e.g., FORCE learning) or hybrid network architecture to train networks of spiking neurons~\cite{sussillo2009generating,depasquale2023centrality}. 
%This often makes them achieve worse results than BP in practice. 
However, ES differs from these approaches in that it is formulated differently to BP through the smoothed loss and estimation.  A consequence is that there are different cases where ES can be comparable or even show performance advantages compared with BP, e.g., for the type of stiff computation problems illustrated in this study.

\section{Methods}

\subsection{Modelling Networks of Biophysical Neurons}

\label{biophysicalneuronmethods}

In this section we describe biophysical neuron models used throughout this work and introduce notation related to such models. In the next Section \ref{casestudy} we provide more detail on the specific choice of neuron model and numerical methods.

A detail model of the biophysics underlying neuronal dynamics is the Hodgkin--Huxley (HH) neuron \cite{hodgkin1952quantative}. HH describes the activity of a neuron's membrane potential voltage (typically denoted as $V$) given by the interplay between different ionic channels that up- and down-regulate voltage activity. Fundamentally, when the neuron is sufficiently stimulated, the neuron is driven to ``spike,'' during which time the voltage demonstrates a rapid increase then drops (see Figure \ref{fig:morrislecar}). Theses spikes can be a source of problems when simulating these neurons, as, depending on the physiology, they can be extremely fast--almost discontinuous--so the state derivative can grow very large. Many simplifications with varying degrees of biophysical realism exist, as detailed below in Section \ref{casestudy}. 

\subsection{Backpropagation for Systems of ODEs}

\label{backpropmethods}

Backpropagation can implemented using \textit{automatic-differentiation}, leveraging the chain rule: operations computed in a forward pass are sequentially differentiated backwards to compute gradients for supervised learning. BP is efficient since it computes derivatives using a combination of a single forward-pass and a single backward-pass. BP is therefore a method that would be advantageous for supervised training with biophysical networks.

A critical difference between the networks we consider and artificial neural networks is that each individual neuronal unit is described by systems of ODEs, making BP harder to adapt. How BP can be extended to systems of ODEs has been a topic of research. A direct approach is \textit{backpropagation-through-time} (BPTT) which relies on the fact that solving an ODE numerically forwards in time comes down to a sequence of operations, so if these operations are differentiable we can backpropgate through them to compute gradients \cite{rumelhart1986learning}. 

An alternative to BPTT is the \textit{Neural-ODEs} method (NDEs), which more explicitly incorporates the fact that we are solving an ODE \cite{chen2018neural,rumelhart1986learning}. NDEs derives an ODE for a network's \textit{adjoint} which can be used to compute gradients. In principle, NDEs does not require storage of forward-pass evaluations, since it only requires the final state of the system to numerically simulate the adjoint backwards in time \cite{chen2018neural}. However, in practice it is common to store the forward-pass evaluations since this has been shown to provide more stable results \cite{kim2021stiff}. In this work, we will refer to NDEs without storing the forward-pass as \textit{full-NDE} (fNDE) and NDEs with cached forward-pass as \textit{partial-NDE} (pNDE). 

An initial step in our work is to determine some scenarios in which theses BP variant break down. It has already been observed that they can break down in a variety of cases -- e.g., exploding/vanishing gradients with BPTT, numerical instability from the adjoint approach in both NDE approaches, and irreversibility of an ODE in the fNDE case \citep{kim2021stiff}. As we observe below, \textbf{biophysical neural networks can lead exhibit similar problems, making them particularly difficult to train.}

\subsection{Evolutionary Strategies}

\label{derivation}

BP is not the only way to approximate gradients for training neural networks. For example, the work of Salimans et al. \cite{salimans2017evolution} demonstrated that sampling-based \textit{evolutionary strategies} (ES) can be a viable alternative to BP. ES possess some desirable qualities including generalization and stability. These advantages come at a cost, as they are naturally less efficient for large problems, since they require many samples to approximate gradients. However, with parallel computing ES can be faster be effectively scaled to even surpass BP as no backwards pass is required \cite{salimans2017evolution}. 

We introduce the ES methodology; for more details, we refer the reader to \cite{salimans2017evolution,mohamed2020monte}. The general idea of ES is to first randomly perturb network parameters for a problem to generate observed losses, then to use these observed losses to approximate the gradient for descent.

% In certain cases BP can often result in completely wrong gradients due to instability and stiffness of the loss landscape. To address these problems, it is desirable to smooth out the local variability and fit a gradient empirically. It's natural to try something like finite differences. However, naive finite differences does not scale well in high dimensions (as is the case for neural networks), scaling linearly with the number of parameters, and can be inaccurate unless we use a higher order finite difference. Evolutionary strategies (ES) are an alternative approach that can be thought of as an stochastic analogue of finite differences, alleviating these problems \cite{salimans2017evolution}. ES operate by taking random samples according to some distribution (usually normal) in the parameter space, centered locally at the current network parameters. Next, the gradient is approximately computed from the sampled losses. Hence, they are similar to finite differences but instead of sampling in specific directions they use randomly drawn directions every time. ES can outperform direct finite differences for large parameter problems since they rely on a Monte-Carlo estimate for gradient calculation, making them scale with the ``effective dimension'' of the loss in parameter space, instead of the number of parameters. We derive the ES approach below for the supervised learning context.

ES treats the network as a black box, $N(\theta; x)$, where $x$ denotes input and $\theta$ denotes the parameters of the network, of some dimension $m$. First, consider the task of minimizing a loss function $L$ on a set $D$ of labeled data pairs $(x,y)$:
\[\underset{\theta \in \R^m}{\text{argmin}} \sum_{(x,y) \in D} L(N(\theta, x), y)\] 
Instead of directly solving this problem, consider the problem of minimizing the "average" loss, averaged over variations in the parameter $\theta$, resulting in a new loss function:
\begin{align}
\label{smooth_loss_def}
L_p(x) &\defeq \mathbb{E}_{v \sim p_\theta}[L(N(v; x),y)] = \int_{v \in \R^m} p_\theta(v) L(N(v; x),y) \text{d}v,
\end{align}
for some choice of distribution $p_\theta$ centered at $\theta$ in the parameter space. We focus on the normal distribution $\mathcal{N}(\theta, \sigma)$, where $\sigma$ is a hyperparameter. Using other distributions, as in variance reduction \citep{mohamed2020monte}, is a possible future direction. Note that if $p_\theta = \delta_\theta$ is a Dirac-delta at $\theta$, then $L_p = L$, so $L_p$ is a generalization of the original loss.  

\EcommentDone{Add a couple of sentences explaining how the optimiz is done ... eg have fixed sigma, and updating the mean.  And state that the reported values in loss curves etc are evaluated exactly at this single mean value (ie are L not $L_p$ ... }

% \begin{wrapfigure}{r}{0.35\textwidth}
% \centering
% \begin{minipage}[t]{0.35\textwidth}\centering
\begin{figure}[t]
    \centering
    \includegraphics[width=0.35\textwidth]{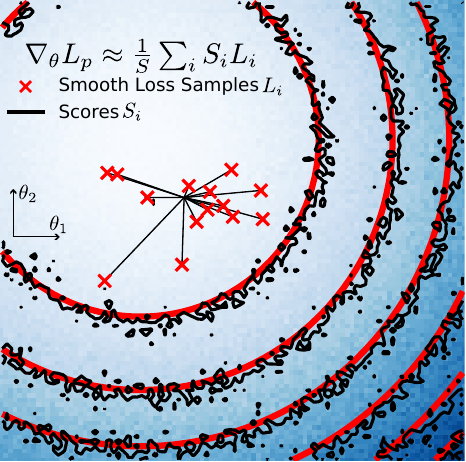}
    \caption{Gradient calculation in two dimensions. The smoothed loss (red contours) is sampled in the parameter space. The gradient is estimated as the mean of the distribution scores, multipled by the sampled losses. In the case where the sample distribution is a normal distribution, scores are simply the vector offsets for each sample.}
    \label{fig:sampling}
\end{figure}

% \end{minipage}
% \end{wrapfigure}

What value does minimizing $L_p$ provide? In many cases, it might be desirable to have a loss that is more robust to parameter variation. However, the primary motivation is that it makes computation of the gradient feasible using Monte-Carlo estimation. In particular, following \cite{mohamed2020monte}:
\begin{align}
\label{derivative_formula}
    \nabla_\theta L_p(x)  & =  \nabla_\theta \int_{v \in \R^m} p_\theta(v) L(N(v; x),y) \text{d}v \\   
     & =   \int_{v \in \R^m} \nabla_\theta p_\theta(v) L(N(v; x),y) \text{d}v \\   
     & =   \int_{v \in \R^m} p_\theta(v) \nabla_\theta \log p_\theta(v) L(N(v; x),y) \text{d}v \\   
	 &= \mathbb{E}(\nabla_\theta \log p_\theta(v)).
\end{align}
Steps 2-3 follow from Leibniz rule, and 3-4 follows from the ``log-trick'' identity:
\begin{align*}
    \nabla p_\theta(v) = p_\theta(v) \frac{\nabla p_\theta(v)}{p_\theta(v)} = p_\theta(v) \nabla_\theta \log p_\theta(v).
\end{align*}
This expected value can be approximated through Monte-Carlo estimation, resulting in:
\begin{align}
    \nabla_\theta L_p(x) & \approx  \frac{1}{S} \sum_{s=0}^S L(N(v^{(s)}; x), y) \nabla_\theta \log(p_\theta(v^{(s)})),
\end{align} 
Thus, we arrive at an estimate based $S$ parameter samples $v^{(s)}$, to be drawn from the distribution $p_\theta(v^{(s)})$ via Monte-Carlo sampling. This expression states that the gradient can be predicted using a weighted sum of evaluations of the network using ``noisy" parameters (see Figure~\ref{fig:sampling} for an illustration). The weights are given by the log derivative, which can be explicitly evaluated in most cases. 

The value of the above estimate (Equation \ref{derivative_formula}) is that Monte-Carlo estimation scales much better with dimension than direct numerical methods for gradient estimates.\EcommentDone{clarify what these direct methods are that scale poorly -- presumably eg finite diff?}
% \footnote{The estimate (Equation \ref{derivative_formula}) can likely be dramatically improved even further using more complex techniques, i.e. from reinforcement learning or variance reduction techniques for score-function estimators \citep{williams1992simple, mohamed2020monte}. In this work, the estimate is used as-is for simplicity and since it is observed to be sufficient in the cases considered.} 
For example, using finite difference to compute gradients scales linearly in the number of parameters. On the other hand, Monte-Carlo estimation scales with the ``effective dimensionality'' of the loss instead of the physical number of parameters. This is particularly advantageous for neural networks since the number of parameters is typically very large and curbing scaling is a central difficulty in estimating gradients numerically.

One advantage of the evolutionary algorithm is the \textbf{filtering of observation noise} when using stochastic network components, making it better suited for physical problems. In particular, suppose
\begin{align} \label{eqn:lossvariability}
    \text{loss}(X) = L(X) + \varepsilon(X),
\end{align}
where $L$ is a well-defined function and $\varepsilon$ models noise or granular local variation, potentially depending on $X$, with mean $0$. Backpropagation, which only considers a single evaluation, may be skewed by the term $\varepsilon(X)$. However, ES relies on expected value, filtering out $\varepsilon(X)$ since it has mean $0$.

Pseudo-code for ES is provided below. The noisify() operation creates $S$ copies of the parameters $\theta$ drawn from $p_\theta$, centered at $\theta$.

\begin{algorithm}[H]
\caption{Evolutionary Strategy}\label{alg:cap}
\textbf{Input:} dynamic parameters $\theta$, input $x$, desired output $y$ \\
\textbf{Output:} approximate gradient, $\text{output} \approx \nabla_\theta L_p(x)$
\begin{algorithmic}
\State $\theta_{\text{noisy}} = \text{noisify}(\theta)$
\State $\text{output} = 0$
\For{$v \in \theta_{\text{noisy}}$}
\State $\text{loss} = L(N(v; x), y)$ 
\State $\text{output} = \text{output} + \text{loss} \cdot \nabla_\theta \log(p_\theta(v)) / S$
\EndFor{} \\
\Return $\text{output}$
\State \textbf{finish}
\end{algorithmic}
\end{algorithm}

In summary, a fixed $\sigma$ is chosen for the normal distribution by hyper-parameter sweeping (we use $\sigma = 0.1$). Next, $S$ random copies of the weights are chosen from the distribution $p_\theta$ and the loss resulting from a forward pass applied to the input batch is recorded for each sampled weight trial (this can be effectively performed in parallel in one large batch; we choose $S = 100$ in the biophysical neural network case). Finally, the sampled losses and noisy parameters are used in equation \ref{smooth_loss_def} to compute approximate gradients. These gradient estimates are then used with an optimizer (we used ADAM) to perform stochastic gradient descent \cite{kingma2017adam}. In practice, ``mirrored sampling'' as in \citep{salimans2017evolution} is used, where half of the weights are random and the other half are set by negating these samples, leading to a substantial decrease in variance. 
\begin{table}[t]
    \begin{tabular}{c|p{6cm}|p{6cm}}
        \hline 
        \textbf{Approach} & \textbf{Pros} & \textbf{Cons} \\
        \hline
        BPTT & More stable since no backward ODE solving. & Memory-intensive.  \\
        pNDE &  Adaptive step size BP. & Memory-intensive.  Adjoint dynamics can cause instability \cite{kim2021stiff}. \\
        fNDE & Constant memory over time. & Can be unstable due to adjoint solve \textit{or} reversing the ODE \cite{yingbo2021comparison}. \\
        ES & No BP, hence high stability.  Constant memory over time. & May need many samples.  \\
        \hline
    \end{tabular}
    \label{summarytable}
    \vspace{10px}
    \caption{Summary of approaches for gradient computation of neural ODEs considered and their potential up- and down-sides.} 
\end{table}

% \subsection{Biophysical Neuron Parameters and Choice of Numerical Method}
% \label{sec:methodsbnns}

% \textbf{Choice of constants TODO}

% \subsubsection{Leapfrog Stepper} 

\section{Results}

\subsection{Case Study 1: Learning to set firing rates in a single biophysical neuron}

\label{casestudy}

To motivate our exploration of training networks of biophysical neurons (BNNs), we begin with the simple case of a \textit{single neuron}. We find that even this simple case result in irregularities causing backpropagation (BP) to diverge, in the sense that gradients blow up or accumulate errors of high magnitude, rendering them unusable for descent. This is seen in all cases of BP mentioned in the prior section: direct BPTT, fNDE and pNDE. We then apply the evolutionary strategy (ES) to demonstrate the advantages it can provide. In summary, we observe that even a single biophysical neuron can be used as a testbed for evaluating the efficacy of different direct supervised learning methods for biophysical problems. %Fitting such neurons is thus an interesting neural ODE problem for future investigation. 

\subsubsection{Choice of Neuron Model}

%\begin{figure}[t]
%    \centering
%    \includegraphics[width=0.7\textwidth]{which.pdf}
%    \caption{\textbf{Choice of neuron model.} Figure taken from Izhikevich \cite{izhikevich2004which}.  %\EcommentDone{Might need to contact publisher?}}
%    \label{fig:whichmodel}
%\end{figure}

In order to ``train'' a single neuron in a supervised learning setup, the choice of (i) neuron model and (ii) loss function need to be specified.  We first motivate the former.

A variety of neuron models exist, differing in multiple ways. These models can be arranged based on increasing biophysical complexity.  Naturally, increases in computational complexity tend to accompany increases in biophysical complexity \cite{izhikevich2004which}.  Leaky Integrate-and-Fire (LIF) cells lie on the simplest end of the range:  they have few parameters and can be efficiently simulated, but lack biophysical realism in their membrane dynamics and spike generating mechanisms.  Hodgkin--Huxley (HH) type models lie at the other:  they explicitly model the dynamics of multiple gating variables, representing ion channel kinetics, in addition to membrane voltage.  A middle ground is occupied by models that use timescale separation or other methods to reduce HH type models to fewer equations. 

%it empirically requires $\sim$1200 FLOPS for a single evaluation and has a complex system of parameters that are difficult to tune to specific desired behaviors.  

The majority of current work training spiking neural networks has understandably focused on the LIF model due to its simplicity.  Here, we take a step from here toward more biological realism by studying one of the popular ``middle ground'' models, the Morris--Lecar (ML) neuron \EcommentDone{add original-ish cite for ML model} \cite{morris1981voltage}.  This model retains, if minimally, the signature of HH type models, in that it explicitly models the dynamics of a (single) gating variable in addition to the membrane.

\begin{figure}[t]
\centering
\includegraphics[width=0.8\textwidth]{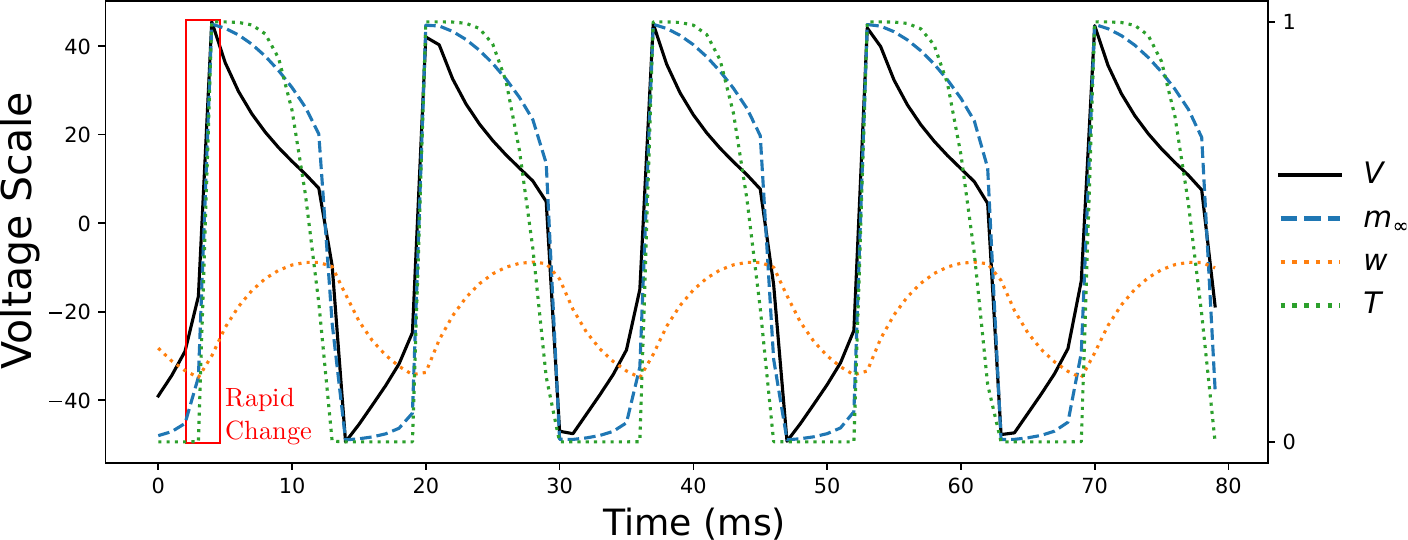}
    \caption{\textbf{The Morris--Lecar Neuron Model.} A snapshot of typical ML tonic firing behavior is shown, demonstrating the interplay between the potassium and calcium gating variables ($m_\infty$ and $w$, respectively), and how they drive the voltage to spike and decay. Also shown is the output function, $T(V)$. Voltage varies according to the left scale and other variables are on the scale (0-1) on the right. The red annotation box notes a ``spike initiation'' region, where the neuron states can change extremely fast depending on model parameters.}
    \label{fig:morrislecar}
\end{figure}

%Figure \ref{fig:whichmodel} illustrates this range (adapted from \cite{izhikevich2004which}). 

%Hodgkin-Huxley (HH) lies at one end of the spectrum: it empirically requires $\sim$1200 FLOPS for a single evaluation and has a complex system of parameters that are difficult to tune to specific desired behaviors.  
% The ''surrogate gradients'' approach for dealing with discontinuities due to spiking of this model has proved extremely effective \cite{zenke2018surrogate}. 

%In this work, the Morris--Lecar (ML) neuron model is used since it is occupies a middle ground: it is more computational efficient and simple than HH, but it is more biophysically relevant than LIF, maintaining the essence of the HH model, which is modelled by the interplay between different ionic channels \cite{izhikevich2004which}. Another model to consider that is both expressive and efficient to simulate is the Izhikevich model \cite{izhikevich2004which}. The Izhikevich model was designed to be expressive, demonstrating a variety of behaviors with a simple ``quadratic integrate-and-fire'' update equation. However, this model still relies on thresholding to give spiking dynamics and is not modelled through ionic channel dynamics. In comparison to the Izhikevich model, since the ML neuron is so comparable to the HH model, almost everything we describe in this work naturally extends to HH and more complex variants. In summary, our primary aim is to minimally simplify the model, using neurons that incorporate a high degree of biophysical realism and are more difficult to simulate and train.

In more detail, the ML model describes a neuron with two dynamical states: a membrane voltage potential, $V(t)$ and a potassium gating variable, $w(t)$ regulating the opening and closing of a potassium ionic channel \cite{morris1981voltage}. Input currents come in the form of a calcium, potassium and leak channel, as well as external applied current:
\begin{align} \label{eqn:morrislecar}
\begin{split}
C \frac{{dV}}{{dt}} &= g_L(V_L - V) + g_{Ca}m_{\infty}(V)(V_{Ca} - V) + g_Kw(V)(V_K - V) + I_{app} \\
\frac{{dw}}{{dt}} &= \phi \cdot \frac{{w_{\infty}(V) - w}}{{\tau_w(V)}}
\end{split}
\end{align}
\begin{align*}
m_{\infty}(V) &= \frac{1}{2} \left(1 + \tanh\left(\frac{V - V_1}{V_2}\right)\right); \, w_{\infty}(V) = \frac{1}{2} \left(1 + \tanh\left(\frac{V - V_3}{V_4}\right)\right) \\
\tau_w(V) &= \frac{1}{\cosh\left(\frac{V - V_3}{2V_4}\right)}
\end{align*}
The choice of constants is given in Section \ref{sec:methodsbnns}. Spiking occurs in this model through the interplay between the three input channels. In the presence of sufficiently high applied current, $I_{app}$, the neuron will exhibit periodic ''tonic'' spiking.  In brief, spikes are driven by rapidly opening $Ca$ channels ($m_\infty \rightarrow 1$), causing the voltage to increase toward the positive value $V_{Ca}$  (see Figure \ref{fig:morrislecar}; $V_{Ca}=120$ mV here) \EcommentDone{put in numerical value, also check in appendix that same notation used with $V_{Ca}$ instead of $V_1$, etc}). Following this, a potassium channel gated by $w(V)$ more slowly opens, driving the voltage activity back toward rest ($V_L = -60$ mV). Note that increasing $\phi$ will speed up transition rates of $w$, making spikes shorter, and shrinking $V_2$ or $V_4$ will make $W_\infty$ or $\tau_w$ change value more quickly, respectively.

\subsubsection{Choice of Task:  Output and Loss Functions}

With the neuron model defined, we turn to the task on which we will train our single-neuron via different learning approaches. Our choice is a simple one, in which we seek to learn a constant input current $I_{app}$ that will drive spiking of the single neuron at a desired rate.  Here, the value of $I_{app}$ may be thought of as either the bias for a one-neuron network, or the weight connecting that neuron to a constant source of input. We choose in particular to minimize firing rates, via a ``\textit{turn off} task,'' in which the loss measures the mean output of the neuron. If $T(V(t))$ is the spiking output of the neuron at time $t$, the loss function is defined as:
\begin{align}
    L := \int_0^\tau T(V(t)) \text{d} t,
\end{align}
for some timeframe $\tau$. 

We turn next define the spike output of the neuron.   Quantifying spikes in terms of voltage can be done in multiple ways \cite{gerstner2014neuronal}; the most common and simplest approach is to quantify spikes based on a threshold, $V_T$. The discrete approach is as follows:
\begin{align}
    T(V(t)) &= \begin{cases} 1 & V(t) \geq V_T \\ 0 & V(t) < V_T \end{cases}.
    \label{eqn:outputdefstiff}
\end{align}
However, this approach is not compatible directly with backpropagation (BP) since it is non-differetiable. This can be easily alleviated by using a soft threshold, i.e. a sigmoid function:
\begin{align}
    T(V) := \frac{1}{1 + \exp(-\frac{V - V_T}{K_p})}.
    \label{eqn:outputdef}
\end{align}
Note such an approach of smoothing the spiking activity is the inspiration for ``surrogate gradients,'' which have been effective for training LIF neural networks \cite{zenke2018surrogate}, and is also used more broadly in computational neuroscience \citep{gerstner2014neuronal}.  In the limit $K_p \rightarrow 0$, the derivative $T'(V)$ approaches a Dirac delta centered at $V_T$. This will naturally cause BP gradients at spike times to be extremely large and unstable. Hence, choosing a reasonably smooth $T(V)$ is important when BP is to be used or compared (here we use $V_T = 10, K_p = 3$). 

%% E edit from here down ... 

We briefly note that smoothness of $T(V)$ is not a concern with ES. Unlike BP, which relies on differentiability of each operation to backpropagate gradients, ES does not propagate gradients and instead samples in the parameter space to numerically approximate the gradient. This enables us to train networks using ES even for output functions  $T(V)$ that would preclude use of BP, such as Equation \ref{eqn:outputdefstiff}.

\begin{figure}[t]
\includegraphics[width=1.0\textwidth]{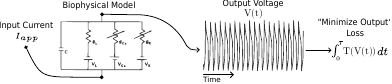}
    \caption{
    \textbf{Case study:  controlling the spike rate of a single Morris--Lecar neuron model.} 
    We simulate a single Morris--Lecar neuron receiving a fixed input current and producing an output voltage over time. The goal is to minimize the total output voltage, and this is done by adjusting the parameter of the input current $I_{app}$. The circuit shown models a single Morris--Lecar neuron and is adapted from \cite{morris1981voltage}.}
    
    % \textbf{A} illustrates the setup of the simple problem. We simulate a \textit{single} Morris--Lecar neuron with a fixed input current and produce an output voltage over time. The task is to "minimize output," i.e. to minimize the total voltage over the time period. To do so, the input current $I_{app}$ is tuned as a learnable parameter since the neuron behavior behaves relatively simple with respect to this parameter. \textbf{B} illustrates the resulting loss landscape in black for this task with varied values of $I_{app}$. Note that the curve is jagged and noisy but has a clear upwards linear trend. This is due to the fact that increasing input current to a neuron increases its rate of spiking, on average. In red is shown the smooth loss after applying a Gaussian smoothing (shaded region). \textbf{C} demonstrates that the derivatives of the loss curves in B with respect to the input parameter $I_{app}$. Note that the original jagged loss gives an unstable and chaotic gradient. Finally, \textbf{D} shows that training this simple task with direct BP does not converge due to instability of the gradient. However, training using the smooth loss gradient trains.  
    \label{fig:motivate}
\end{figure}

\begin{figure}[t]
\centering
\includegraphics[width=0.7\textwidth]{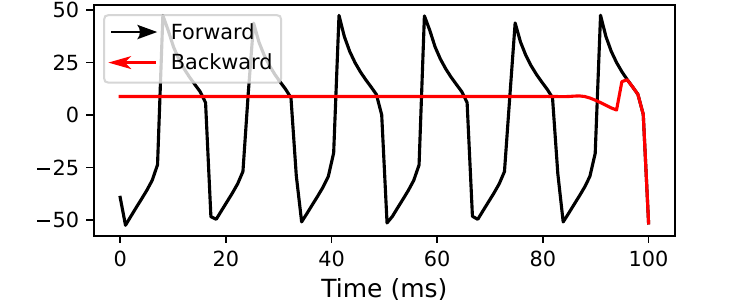}
    \caption{\textbf{Irreversibility of Morris--Lecar, making the ``full'' neural ODEs approach not viable.} Simulation is performed forwards and backwards in time using the leapfrogging scheme mentioned in the text with $\delta t = 0.001$ ms. However, the backwards evaluation rapidly diverges almost immediately:  the red backward curves tracks the forward solution only from time 100ms to roughly 95ms. The rapid changes of voltage at spikes likely contributes to this small window of reversibility (see Proposition \ref{prop1} and \ref{prop2}) \EcommentDone{state somewhere that the smallness of this window persists under other numerical schemes you've tried?}.}
    \label{fig:nonreversible}
\end{figure}

\subsubsection{Identifying Pitfalls BP Encounters with Biophysical Models}

\label{sec:pitfalls}

With the model and loss defined, we turn to the problem of learning the parameters that solve the task at hand (i.e., minimize the loss). In this case, there is a single learnable parameter $I_{app}$ (see Equation \ref{eqn:morrislecar}).  We initialize this with $I_{app} = 100$.  For the simple \textit{turn off} problem, there is a clear solution, which is to decrease $I_{app}$ so that the neurons stop persistently spiking.  Nevertheless, with certain parameters or a sufficiently long timeframe, BP is often unable to find it, as it suffers from multiple problems in computing a correct gradient direction for descent: 

%The value of this setup is it gives us a simple problem involving a single ML neuron where the solution is simple to understand and compare against. 
    
\textbf{(I) Noisy loss curve (poorly defined gradients):} As in Figure~\ref{fig:landscape-train}.A, the ``loss landscape'' for this problem has a clear (upwards) trend, but is noisy and involves large jumps. This happens because the loss function, as simple as it is, quantifies spiking activity -- which is, even after averaging over a timeframe \EcommentDone{are we smoothing over time?  if we're referring to T(V), say smoothing over voltage?}, an essentially discrete phenomenon (we also see large jumps due to bifurcation properties of the single neuron; e.g., Figure \ref{fig:suppzoom}). \EcommentDone{A good place to discuss those ``inset'' traces from peak and valley of loss} In order to better capture the structure of the loss, we need to ``spatially smooth,'' i.e. filter out variability in the parameter space  (see discussion and Equation \ref{eqn:lossvariability}) \EcommentDone{rephrase ... smooth over parameters} .
%or run many trials with different neuronal initial conditions of noise inserted. 
This is the principle behind the evolutionary strategy, dealing with a loss that is averaged locally in the parameter space, as over the distribution of parameters shown in Figure \ref{fig:landscape-train}.A and B (see Section \ref{derivation}). 
    
\textbf{(II) Irreversibility of the ODE:} As in Section \ref{backpropmethods}, one approach to backpropagate gradients with ODEs is ``full neural ODEs'' (fNDEs). As described, this method uses constant memory over time as it is based on the final ODE state after the forward pass, and simulates this backwards through time to compute gradients with the NDEs approach \cite{chen2018neural}.  Full neural ODEs are an important alternative to BPTT and ``partial NDEs'' (pNDEs, storing intermediate \EcommentDone{rephrase along lines of ... storing value of neural states at each timestep} neural states at all forward passes) for biophysical neural networks.  This is because BPTT and pNDE methods, while avoiding the irreversibility issues described next, can quickly become infeasible due to memory issues:  consider that to simulate seconds of the neural activity, we may need thousands to hundreds of thousands of forward ODE steps  \cite{kim2021stiff}. As tasks solved by biological organisms involve many neurons for many seconds or even minutes or more \cite{pei2022neural}, this is a serious problem; the fact that the biophysical models we consider require small timesteps and involve multiple intermediate variables for each timestep evaluation further compounds it. Real-time recurrent learning (RTRL), which utilizes an alternative gradient factorization to BPTT, would also be problematic due to its poor scalability with respect to the network size~\cite{marschall2020unified}.  \EcommentDone{explain a bit more -- add a phrase or short sentence explaining the scalability issue, and clarifying (if I have this right of course) that it's different matter than storing all timesteps ... ?  }

This said, the fNDEs approach delivers its own challenges for models, such as spiking neuron systems, with rapidly changing dynamics. Specifically, multiple works have shown that ``reversing the ODE'' to flow backward in time, as for the fNDE approach, can cause gradients to drift and be unstable \cite{kim2021stiff, gholami2019anode} \EcommentDone{please check to make sure my rephrasing didn't mess this up}; many ODE models are simply irreversible for long enough time windows. We show in Figure \ref{fig:nonreversible} how the Morris--Lecar model suffers from this irreversibility.

% \EcommentDone{I think would be good to be explicit here about why BP avoids this problem (using forward pass or whatever is right thing to say!).  Also, we should make sure the intro here meshes with what we end up doing for ML network (ie, if we just use BP there, we should make it clear that it IS possible just slow / huge memory ...). :)}

To understand how irreversibility might arise, let's consider a simple one dimensional ODE $\frac{\diff x}{\diff t} = f(x(t))$. The reversed ODE is $\frac{\diff x}{\diff s} = -f(x(s))$. While $\frac{\diff x}{\diff t} = f(x(t))$ with locally Lipschitz continuous $f(x)$ is reversible in theory, in practice, there are several complications due to instabilities and numerical errors~\cite{gholami2019anode}. As explained in~\cite{gholami2019anode}, consider a simple linear ODE $\frac{\diff x}{\diff t} = \lambda x$. If $\lambda < 0$, reversing the ODE flips the sign of the derivative of $f$, thereby amplifying the errors exponentially fast while solving the reverse ODE. \textbf{Exponential amplification of numerical errors when solving the reverse ODE can lead to the accumulation of numerical drift, especially during long-duration simulations, rendering it impractical to solve the reversed ODE.} The following propositions indicate how this issue can apply to biophysical neuron models, starting from LIF (Proposition 1) and extending to BNNs (Proposition 2). 
% Helena: using the paragraph above instead to elaborate a bit more on what is the irreversibility issue and how it might arise, in case if readers aren't familiar with it
%Note that the irreversibility issue becomes even more of an issue when considering networks of neurons where each neuron's precise spike timing can influence the entire system. The following proposition speaks to the irreversibility of the biophysical neural models we consider here:

\begin{prop}
    Considering a leaky-integrate-and-fire (LIF) model, $\tau \frac{\diff V}{\diff t} = E_L - V$, with voltage $V$, time constant $\tau > 0$, and reversal potential $E_L$. The derivative of the reverse ODE function is strictly positive.
    \label{prop1}
\end{prop}

\begin{proof}
    The reverse ODE is simply
    \[
    \frac{\diff V}{\diff s} = -\frac{E_L - V}{\tau},
    \]
    with the derivative of the reverse ODE function as 
    \[
    \frac{\diff (\frac{\diff V}{\diff s})}{\diff V } = \frac{1}{\tau} > 0
    \]
\end{proof}

\begin{figure}[t]
\includegraphics[width=1.0\textwidth]{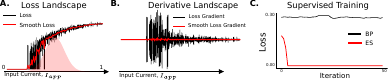}
    \caption{
    \textbf{Loss landscape and learning on this landscape for the single neuron ``turn-off'' task.} 
    \textbf{A} shows the resulting loss (in black) when different values of input current are used. Large jumps in loss are due to bifurcation behavior where neuron transitions from no spiking to tonic spiking (see Figure \ref{fig:suppzoom}). \EcommentDone{Would be great to add an inset or appendix fig with traces at a ``peak" and ``valley" value of $I_{app}$, so readers can understand what is driving the big jumps, and to discuss that in a sentence or two in text ...} The graph is noisy but shows a general upward trend, indicating that a higher input current increases the neuron's spiking rate. The red line shows the smoother loss after Gaussian smoothing. Note that all trials used the same initial condition and 10 seconds of simulation with high firing, hence the high variability in loss. \textbf{B} presents the derivatives of the loss curves from B, highlighting that the original rough loss leads to an unstable and irregular gradient. Lastly, \textbf{C} demonstrates that directly training this model does not converge due to the gradient's instability. However, it does converge when training uses the smoother loss gradient using the evolutionary strategy (ES).  \EcommentDone{Wonder if we should split up this Fig.  Thinking that early panels are (or could be) referenced from text pretty early -- well before current Fig 5, and panel D then later.  So could just make panel D a one-panel fig that comes later ...} }
    
    % \textbf{A} illustrates the setup of the simple problem. We simulate a \textit{single} Morris--Lecar neuron with a fixed input current and produce an output voltage over time. The task is to "minimize output," i.e. to minimize the total voltage over the time period. To do so, the input current $I_{app}$ is tuned as a learnable parameter since the neuron behavior behaves relatively simple with respect to this parameter. \textbf{B} illustrates the resulting loss landscape in black for this task with varied values of $I_{app}$. Note that the curve is jagged and noisy but has a clear upwards linear trend. This is due to the fact that increasing input current to a neuron increases its rate of spiking, on average. In red is shown the smooth loss after applying a Gaussian smoothing (shaded region). \textbf{C} demonstrates that the derivatives of the loss curves in B with respect to the input parameter $I_{app}$. Note that the original jagged loss gives an unstable and chaotic gradient. Finally, \textbf{D} shows that training this simple task with direct BP does not converge due to instability of the gradient. However, training using the smooth loss gradient trains.  
    \label{fig:landscape-train}
\end{figure}

\begin{figure}[t]
\includegraphics[width=1.0\textwidth]{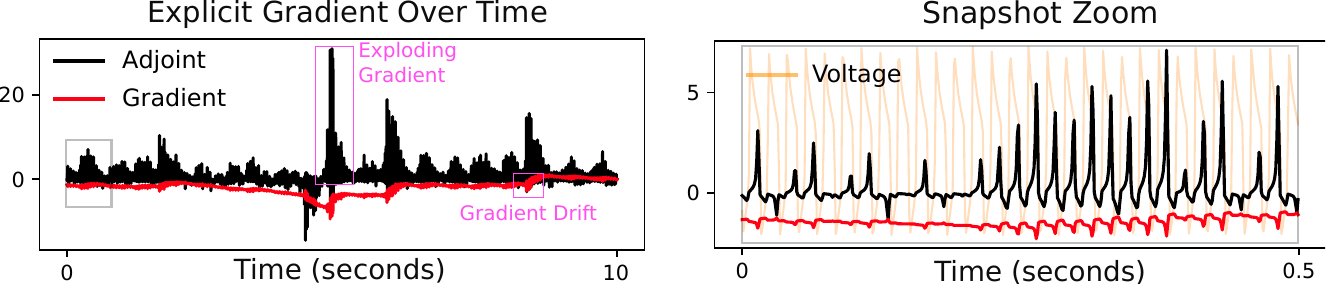}
    \caption{\textbf{ODEs Gradient Calculation Demonstrates Drift and Exploding.} Gradients are calculated backwards through time via the adjoint, using the stored forward states approach (backwards evaluation does not work, as in Figure \ref{fig:nonreversible}).  Results are shown for the single neuron \textit{turn off} benchmark described in the text with a 10 second timeframe. Right panel shows a zoom in of the first 0.5 seconds with voltage activity in orange. The running gradient (from right to left) is shown, as well as the ``adjoint,'' which measures the derivative $\frac{\partial L}{\partial V(t)}$ at each timestep $t$. Importantly, note that both are highly correlated with the beginning and end of voltage spikes: the adjoint itself appears to ``spike'' very quickly at these times. We can see that the gradient calculation drifts as it is computed backwards and suffers from ``exploding'' at multiple places. Note, as in the text and Figure \ref{fig:landscape-train}.A, the trend in loss as the applied current increases is positive, but the gradient drifts to negative values (as shown) over the time window.}
    \label{fig:adjointdrift}
\end{figure}

Therefore, even in the simple LIF case, reversing the ODE can potentially lead to the amplification of numerical errors, with the potential amplification inversely proportional to the time constant $\tau$. This problem should extend to other conductance-based models, which also include similar terms with channel voltages being pulled to their reversal potentials by their respective "time constant". We also demonstrate that this issue applies to the gating variables in these models (e.g., Equation \ref{eqn:morrislecar}).

\begin{prop}
    Consider a model with voltage variable $V$ and gating variables $m_i$ ($i=1, 2, ...$), with the dynamics of each gating variable provided by $\frac{\diff m_i}{\diff t} = \frac{m_{i,\infty}(V) - m_i}{\tau_{m_i}(V)}$ for some $V$, where time constant $\tau_{m_i}(V) > 0$ and steady state gating variable value $m_{i,\infty}(V)$. The partial derivative of the reverse ODE function for the gating dynamics with respect to each gating variable is positive.
    \label{prop2}
\end{prop}

\begin{proof}
    The reverse ODE of $\frac{\diff m_i}{\diff t} = \frac{m_{i,\infty}(V) - m_i}{\tau_{m_i}(V)}$ is 
    \[
    \frac{\diff m_i}{\diff s} = -\frac{m_{i,\infty}(V) - m_i}{\tau_{m_i}(V)},
    \]
    with the partial derivative of the reverse ODE function with respect to $m_i$ as
    \[
    \frac{\partial (\frac{\diff m_i}{\diff s})}{\partial m_i } = \frac{1}{\tau_{m_i}(V)} > 0
    \]
\end{proof}

Due to positive derivatives of the reverse ODE function, numerical errors may be amplified when solving the reverse ODE. This amplification is exacerbated especially if the derivative is high. This amplification can limit the accuracy of backpropagation when using the adjoint-based method, potentially causing an accumulation of numerical drift, as observed in our reverse ODE solution (see Figure \ref{fig:adjointdrift}). To enhance interpretability, we have chosen to focus our discussion on one-dimensional ODEs rather than a coupled system of ODEs, as would be the case in Equation \ref{eqn:morrislecar}. We remark that the irreversibility issue could become more pronounced when considering networks of neurons where each neuron's precise spike timing can influence the entire system.

\textbf{(III) Exploding gradients at spikes:} We note that irreversibility \EcommentDone{got a bit confused about use of phrase gradient drift ... do we mean the irreversibility studied in previous point?  if so can we just call it irreversibility when we refer to the inability to do a fNDE approach?  sorry if i'm confusing things!} is not the only problem caused by the high and positive derivative of the reverse ODE function. In particular, the voltage and activation gating variables can rapidly change at spike initiation, causing behavior that is ``almost'' discontinuous. In the proposition above, the time constant plays a crucial role in determining the magnitude of the voltage rate of change. Rapid changes indicate high voltage time derivatives, which in turn can lead to gradient explosions. Note that this occurs irrespective of whether we use the neural ODEs or direct BPTT: both require gradient computation at spike times. As noted below, the problem can be compounded when considering a network of neurons, causing gradients to blow up. Figure \ref{fig:adjointdrift} illustrates the exploding of gradients in certain cases focused around spikes. Exploding gradients can cause the gradient to behave unstably and errors can gradually accumulate over time, causing the gradient to drift (for example the sign can be completely wrong, as in Figure \ref{fig:adjointdrift}).   

\subsubsection{Applying Evolutionary Strategies}

In comparison to BP, evolutionary approaches perform in more general contexts on this simple benchmark (see Figure \ref{fig:landscape-train}.C). As in the Methods above, ES works by locally sampling in the parameter space (which is one dimensional in this case) and fitting a gradient using Monte Carlo estimation. In essence, it is a gradient estimate for the ``smoothed loss,'' and hence is more robust to noise (see Figure \ref{fig:landscape-train}) \cite{mohamed2020monte}. When we use a long timeframe (10 seconds), direct BPTT gives incorrect gradients with an incorrect sign or large magnitude, while ES is able to capture the general gradient trend and tune the $I_{app}$ parameter effectively (Figure \ref{fig:landscape-train}.C). \textbf{(I)} is alleviated since ES ``filters out'' noise in the loss landscape by sampling multiple times. \textbf{(II)} and \textbf{(III)} also do not pose a problem since evolutionary algorithms are forward pass only, circumventing the problems related to irreversibility when backpropagating backwards. A clear downside of ES is the need for potentially many samples. In the one-dimensional case here, we find that ES has a high variance (i.e. high variability from a good gradient estimate) and thus requires many samples (around 100). \EcommentDone{add quick of phrase on what high var means} \EcommentDone{state here how many samples that was, or refer to where that is given?}Other evolutionary approaches (e.g., random search/weight perturbation \cite{jabri1992weight}) may perform better in this specific example. However, as observed elsewhere \cite{salimans2017evolution}, ES often scales better with dimension since it relies on Monte-Carlo estimation and can leverage efficient gradient descent approaches such as ADAM and momentum-based gradient descent \cite{kingma2017adam}.

\subsection{Biophysical Neural Network Solving Discrete Evidence Integration}

In this section, we evaluate the viability of BP and ES for a task that involves training a recurrent network of biophysical neuron models to solve discrete evidence integration. Evidence integration is one of the most common tasks in computational and experimental neuroscience \citep{wang2002, shadlen2001theory, carli1989distinct,aitken2023,bogacz2006physics} \EcommentDone{Add cite to Bogacz et al, Psych. Rev., 2006}. The crux of the problem is decision making over time: a network is given a stream inputs, each giving weak evidence for or against a decision alternative, the network accumulates (i.e., integrates) these inputs to make an accurate overall decision. There are multiple variants, including discrete, contextual or continuous integration. In this work, we focus on the discrete integration task; 
%instead of continuous integration since it \textit{1.} cannot be easily solved by a simple random reservoir (most biophysical neurons naturally behave like integrators, so random network configurations often have a high chance of solving the continuous integration problem), and \textit{2.} requires the network to identify and translate from the discrete evidence vectors (which are random binary vectors of length 50) to something the network responds to. A 
a schematic overview is given in Figure~\ref{fig:integrator_task}.A.

\subsubsection{Problem Setup and Network Definition}

\label{integrator_task}

\begin{figure}[t]
    \includegraphics[width=1.0\textwidth]{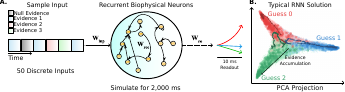}
    \caption{\textbf{Evidence Integration Task Setup and Typical Solution.} \textbf{A} Illustration of task and setup. Inputs are sequences of "evidences" encoded as binary vector inputs. These are fed into a recurrent neural network of biophysically based Morris--Lecar neurons. The neurons' outputs are fed into a linear layer with three task outputs. The mean output over a final short time window (10 ms) is computed and the prediction for which evidence occurs most often corresponds to the output with highest mean readout (1, 2 or 3). \textbf{B} A typical ``integrator'' solution for this problem \cite{aitken2023} showing the PCA projection of the hidden neuron activities over time for an RNN solution, colored by the ground-truth label. Note that the solution is expected to form attractor structures: the hidden state over time is attracted to one of three states quantifying the output decision. }
    \label{fig:integrator_task}
\end{figure}

\begin{figure}[t]
    \centering
    \includegraphics[width=0.9\textwidth]{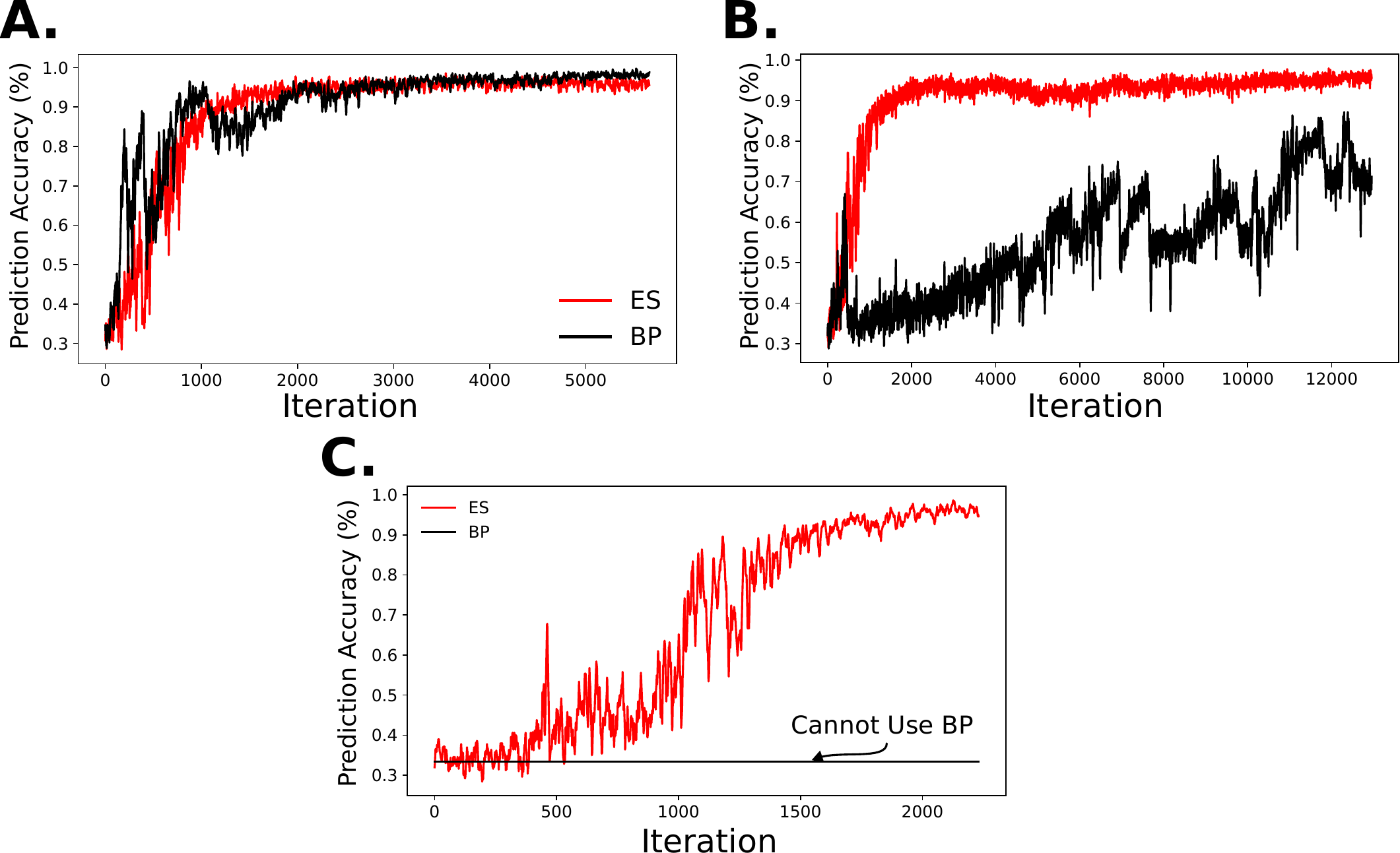}
    \caption{\textbf{Learning the evidence integration task in a network of Morris--Lecar neurons.} \textbf{A-C} Accuracy curves for ES and BP averaged over 5 training runs with randomized initialization for this task. Each iteration uses a batch size of 64 samples 100 samples are drawn from the parameter space for ES. \textbf{A} corresponds to ``non-explosive'' parameters, as in text. \textbf{B} corresponds to the explosive parameter case. In this case, large gradients cause instability training directly with BP, leading to slower convergence than ES. Finally, in case \textbf{C} the memory consumption with BP is too large for training to complete (a 1 second timeframe is used).  \EcommentDone{numbers on axis are a bit hard to read, let's make 'em a bit bigger ...} \EcommentDone{Question here -- could we explain in a sentence how iterations are compared?  ie how many ``runs'' go into each, say?  I know this is not the only important thing since BP requires expensive backward passes and memory, and ES is likely parallelized on GPU etc, but might help readers to give a bit of detail say in figure caption?}
}
    \label{fig:integrator_task_train}
\end{figure}

We use the Morris--Lecar (ML) model of individual neurons, as in the previous example, but now employ multiple ML neurons  arranged in a recurrently connected network. The aim is for the network to integrate evidence and decide which of three possible input types occurs most often in its input stream (Figure~\ref{fig:integrator_task}.A.). At each timestep, an input  arrives as part of this stream:  this input is either of type 1, 2, or 3, or a ``null'' input which indicates nothing.  Each of these inputs enters the network as a different fixed random binary vector.  \EcommentDone{quick clarification -- is the null input also a random vector? or is just zeros?  please fix phrasing here as needed as I might have messed it up :); A: null is a random vector also, not zero.} 
%as well as a ``null'' evidence . 
We feed in a stream that is 50 inputs long for each training sample, with a new input arriving every 6 ms\EcommentDone{fill in xx}; after this, the network must output a decision of in favor of type 1, 2 or 3. To quantify this decision, the mean output of the network for the last 10 ms is integrated and the highest output is chosen. The output of the network is computed by feeding the hidden states through an affine transformation, $W_{out}$. In particular, the loss is given brequiringy
\begin{align}
    L := \int_0^\tau \| W_{out} T(V(t)) - T^* \| dt,
\end{align}
where $T^*$ is the one-hot encoded desired output (a vector in $\R^3$ for the ternary-integration problem) and $T$ is given by Equation \ref{eqn:outputdef}. 
\EcommentDone{Can we clarify the def of the output in the last couple of sentences?  eg is there a one-hot (3 unit) kind of output, or another setup?}  \EcommentDone{Likewise, here or below, would it make sense to write out loss more explicitly?  seems would give us some symmetry with the single neuron turn off case above ...}

The hidden state is the voltage of each neuron $V(t)$. To implement coupling between neurons, we directly feed the recurrent weighted output at time $t$, 
\[z(t) = W_{rec} T(V(t)),\]
back into the network as an input current, i.e. adding $z(t)$ to $\frac{d V}{dt}$ in Equation \ref{eqn:morrislecar}, along with the external current, $I_{app}$, which we fix in this case. We note that other options are available: for example, one could indirectly feed the input through an additional ODE or filter variable influenced by $z(t)$, representing synaptic dynamics. Adding this extra ODE with gradual synaptic dynamics should not affect the considered BP variants or ES, so in this work we opted for the form above.  
%In our case, we used the current-based coupling, which is likely more comparable to gap-junctions instead of synaptic connections in the brain. 
To efficiently simulate neurons, we used a two-dimensional implicit trapezoidal rule for numerical stepping (see Section \ref{sec:methodsbnns} for details).

% In particular, the input is a sequence of ``evidence'' vectors that are random 50-dimensional binary vectors. There are four types of evidence: evidence 1-3 and neutral or ``null'' evidence. The goal of the task is to quantify which type of evidence occurs most frequently in the input sequence. Such tasks have been used extensively in computational neuroscience as a simple way to quantify the ability of a neural circuit to store ``working memory,'' i.e. knowledge that is on a longer timeframe than that of the individual neurons \citep{aitken2023,wang2002,gold2007neural}. In our case, each sequence has length 50 and each input vector is distributed over 4 ms, so the network of ML neurons is simulated for 200 ms in total, roughly consistent with experimentally measured reaction time for this sort of task \citep{gold2007neural}. To efficiently simulate neurons, a two-dimensional implicit trapezoidal rule is used for numerical stepping (see Section \ref{sec:methodsbnns} for details).

\subsubsection{Evolutionary Strategies are More Robust than BP for Training Biophysical Network}

We now evaluate the performance of backpropagation vs parameter sampling-based evolutionary learning strategies (BP vs ES) in training the network of biophysical neurons to solve this task.  In brief, we will show that while there are situations in which BP is effective, ES is consistently more robust and always gives comparable results to BP.  We will illustrate how BP converges successfully and can be more efficient than ES when (1) the timeframe of evidence integration is short and (2) the dynamics of each neuron are \textit{non-explosive} in the way we define below.  However, away from this regime, BP is quick to break down;  ES offer a more robust convergence in each situation.
\vspace{-0.5cm}
%and seem to perform comparable to BP when it does work. 
\begin{enumerate}
    \item[\textbf{Case A} Non-explosive short timeframe:] First, we discuss the situation where BP works effectively, as in Figure \ref{fig:integrator_task_train}.A. We use 128 hidden neurons, a 300 ms timeframe, and \textit{non-explosive} network parameters: i.e. the parameters $\phi, V_2$ and $V_4$ in Equation \ref{eqn:morrislecar} are chosen so that spikes are less frequent, sharp, and short (specific details can be found in Section \ref{sec:methodsbnns}). We show plots of task accuracy vs. training epoch averaged over 5 runs using BP and ES in Figure \ref{fig:integrator_task_train}.A.  The ADAM optimizer was used to follow loss gradients in both cases. BP converges comparably fast as ES, and both achieve a similar final accuracy greater than 95\%. %Note that it would not be surprising if BP converges faster than ES in this case where gradients are reasonably well behaved: when BP works, it is usually more efficient than gradient approximation approaches. 

    \item[\textbf{Case B}] \textbf{Explosive-parameters:} When we modify the parameters mentioned so that the neurons have sharper spikes and spikes occur more rapidly and over shorter durations,  
    %This corresponds to at least some of the diverse settings observed biologically.
    %:  It is desirable for the training method to be robust to this sort of variation of parameters: it is well known that neurons can exhibit extremely fast and varied dynamics in different contexts experimentally; training of such neurons is desirable. 
    the sharp transition to spiking and high spike rate\EcommentDone{can we clarify:  is this sharp T(V), or high frequency of spiking} causes the ``gradient explosion problem'' mentioned above (Section \ref{sec:pitfalls}) \EcommentDone{refer more specifically to material above?}, in which training becomes unstable. As mentioned there, sharp spikes induce ``almost'' discontinuities that render the neural ODEs approach for backpropagation non-viable and causes BPTT to give explosive or nonsensical derivatives. This said, the ADAM optimizer is built to accommodate complex loss structures due to mechanisms such as momentum that are able to adjust the learning rate based on relative gradient scaling, hence the networks does eventually train with BPTT, albeit much slower than with the ES approach (Figure \ref{fig:integrator_task_train}.B). Problems due to explosiveness of neural parameters are expected with complex neuron models which can exhibit a variety of different timescales (e.g., refractory periods vs spikes) manifested through internal dynamics.  
    %which can handle such instability gradually by scaling learning rates \EcommentDone{make sure this will be clear to reader -- does scaling alone solve problem?}, training does eventually occur, 
    %Perhaps some variant of surrogate gradients can be extended in this case, however it is not as clear since there are multiple internal parameters of the ML model that can independently cause problems: $V_2, V_4, \phi$, for example. 
    %Note that BP does gradually converge, likely due to ADAM regulating the large gradient norm gradually, 
    %but it is much slower than via the evolutionary strategy. 
    
    \item[\textbf{Case C}] \textbf{Long timeframe:} Another reasonable case we considered was that of long timeframes. In particular, non-explosive parameters $\phi, V_2$ and $V_4$ were chosen to model individual neurons, but longer timeframe of 1 second instead of 300 ms \EcommentDone{was it 300 ms above?}was used for the integration task (with 50 evidence inputs). In this case, BP (through BPTT or pNDEs) becomes infeasible\EcommentDone{specify the variant of BP -- do we mean BPTT and fNDE?}, since the memory consumption was too great (over 25 GB) to store the full forward pass evaluation. Furthermore using the ``memory free'' fNDE approach is not viable due to the ODE irreversibility problem mentioned above (Section \ref{sec:pitfalls}). A relatively large $\Delta t$ of 0.1 ms was used, enabled by a stable leapfrogging method (see Section \ref{sec:methodsbnns}), but 10,000 timesteps are still required with 4 state variables per neuron \EcommentDone{fill in :)} to capture the full forward pass.
    
    % The alternative implemention of BP is the \EcommentDone{do we need that adjoint-based phrase?} ``memory free'' neural ODEs approach, which does not store the the network state at each forward timestep, but instead solves the underlying ODE  backwards in time for BP.  However, as we illustrated in Figure \ref{fig:nonreversible} via the numerical issues explained in Propositions~\ref{prop1} and~\ref{prop2}, this approach is difficult to use.  This is because the underlying (Morris--Lecar) ODEs for each neuron rapidly fail to be reversible -- here, even for extremely small timeframes of ~10 ms (only 1\% of the full time window). 

    In contrast, ES only requires the final output of the network to guide parameter updates, so longer timeframes have no effect on the memory consumption. Figure \ref{fig:integrator_task_train}.C shows that they converge in this case.  
    % EAs are more restricted by number of samples required (although lower number of samples usually will converge, just slower).   
    \EcommentDone{can we give some more details in last sentence, or remove it?} These second(s)-scale time periods of evidence integration match the timeframes probed in some experimental settings, and certainly within the domain of natural behavior.  Moreover, similarly long timeframes readily arise in other task settings, from modeling more complex behaviors to matching neuronal spike recordings \EcommentDone{maybe cite a paper for the latter} \citep{pei2022neural}. 
    
    % \item Finally, in case D we investigate ``over-parameterizing'' the network with the non-stiff physiological parameters and 200 ms timeframe. It's shown that 128 neurons can solve the task, as in A. Over-parameterizing just means increasing the number of hidden neurons on this solvable task. D used 500 hidden neurons and shows BP can breakdown in this situation. We hypothesize that this is due to BP not being able to resolve the gradients precisely enough due to high spike interaction count. Over-parameterization has been suggested as an effective avenue for use of EAs since they often scale with the ``effective dimension'' of the problem instead of the actual number of physical network parameters, so increasing the latter does not increase the number of samples needed with the EA. 
\end{enumerate}

In summary, we find ES can be effective as an alternative to BP for gradient estimation and training in the evidence integration task.  In particular, ES is robust in the presence of rapid or ``explosive'' dynamics in individual neuron models (which occur in various experimental contexts), as well as long task timeframes.

\subsection{Broader Applicability of ES to Other Neural ODEs}

In the next two sections, we focus on the more broad context of neural ODEs (also referred to as continuous time neural networks) and the applicability of ES. We refer to neural ODEs as neural networks where, instead of a sequence of discrete time steps, results are produced by an ODE solver, with potentially adaptive step size. In other words, the neural network $NN$ describes an autonomous ODE:
\begin{align}
    \frac{\text{d}}{\text{d} t} x = NN(x),
\end{align}
and the final state is given by adaptively stepping through time using the neural network. \EcommentDone{Suggest that we add a sentence or two explaining what the neural ODEs context here is ... i.e., different from does each individual neuron following a more complex ODE model, have vanilla single neuron models that together produce ODE output, if that's right, etc ... } BNNs are therefore a sub-class of neural ODEs. Inspired by the success of ES for the BNNs, we aim to investigate their applicability for more general neural ODE problems. To the best of our knowledge, evolutionary algorithms have not been substantially analyzed and applied in this context as an alternative to backpropgation. Naturally, methods like finite differences have been applied and have similar advantages as ES poses above, including stability (e.g., no irreversibility concerns) and robustness over BP. 

\subsubsection{ES are Effective for Training Over-Parameterized Neural ODE}

\label{neuralodes}

% \textbf{KEY TAKEAWAYS: Lower \# of samples than FD, comparable speed/accuracy, etc. Don't need MNIST or NDEs specifically.}

% First, we study the efficacy of our method versus naive sampling approaches (finite difference, FD), the state-of-the-art biologically plausible learning rule (modprop, MP), and backpropagation (BP). Our method was first applied to a benchmark in neural ODEs (NDEs) in which a two-variable flow is learned using a continuous time neural network with one hidden layer. By scaling up the hidden layer size, we could effectively measure how our method scales with higher numbers of parameters. \nl
% Results with 1,000 hidden neurons and 2 input and output neurons are shown in Figure \ref{fig:loss_ndes} for 7 training runs. We trained with the ES method using 500 samples and a normal distribution with standard deviation $0.01$. Surprisingly, our method appears to outperform vanilla BP in this case. This could suggest that minimizing the smoothed loss, $L_p$, leads to better final results for this particular NDE problem. \nl
% Figure \ref{fig:samples_ndes} demonstrates how many iterations were required with our method to reach the final loss attained by BP. The blue line is the number of samples required for FD. 

\begin{figure}[t]
    \centering
    \includegraphics[width=\textwidth]{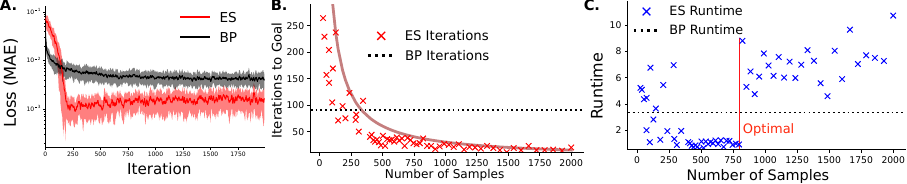}
    \caption{\textbf{ES solves a two-variable neural ODE flow problem and demonstrates effective sampling in large parameter spaces.} A network with 2 inputs and outputs and variable number of hidden neurons is used. \textbf{A} Loss comparison between our method (ES) and BPTT over 7 runs. Shaded regions show standard deviation. 150 parameter samples are used for ES training with 1000 hidden neurons. \textbf{B} Measurement of number of iterations ES takes to reach a final BP accuracy of 0.04 MAE (mean absolute error). The black dotted line shows the number of iterations for BP. Each point is the best result with a varied sample standard deviation for ES. \textbf{C} Same as B but showing effective implementation runtime to converge.}
    \label{fig:nde_acc}
    \vspace{-0.5cm}
\end{figure}

In this section, we apply ES to a non-stiff neural ODEs problem using a continuous-time neural network with a single hidden layer. The task involves fitting a two-dimensional ODE.  Empirically, this task can be effectively solved using as few as 50 neurons in the hidden layer with BP 
% (see Appendix \ref{detailssims} for example trajectories which illustrate the ``flow'' of the ODE)
.  First, we confirm that, in the same setting with 50 hidden neurons, ES achieves essentially the same final loss as BP and with similar or faster convergence rates (see Appendix \ref{detailssims}). 

We then increased the number of hidden neurons to 1000 and repeated network training with both ES and BP.  While it may appear that increasing the number of hidden neurons beyond the 50 is not required for the task, this experimental setup allows us to demonstrate how ES performs in such an over-parameterized setting, giving insights into how this approach scales with number of network parameters. Surprisingly, we find that ES can efficiently train a network with 2,000 parameters using as few as 20 samples (Figure \ref{fig:nde_acc}.B,C), which contrasts with BP which scales linearly with parameter count. 

Taken together, these results indicate that ES can train the network comparably to BP with a low number of parameters and can outperform BP when using a larger number of parameters (Figure \ref{fig:nde_acc}.A-C). This observation suggests that using the smoothed loss might become increasingly advantageous in large network settings; a mechanism for this could be the smoothed loss ameliorating especially high levels of variability resulting from small changes to parameters in these settings.

To further quantify the scalability of our method, we illustrate convergence to the final achieved loss for both BP and for ES with different numbers of samples.  We demonstrate this both for number of iterations required (Figure \ref{fig:nde_acc}.B) and effective runtime (Figure \ref{fig:nde_acc}.C). The latter plot shows that there is an optimal choice of sample numbers for ES which, in this implementation, achieve a faster overall runtime than BP. Notably, the evaluation is done on the GPU in parallel such that using more samples is not necessarily significantly slower. Overall, these results demonstrates that ES can converge efficiently with a relatively small number of samples for non-stiff neural ODEs.  % compared to naive numerical estimation of gradients, e.g., finite difference.

\subsubsection{ES on Stiff Neural ODE Task} 
\label{stiff}

\begin{wrapfigure}{r}{0.4\textwidth}
\centering
\begin{minipage}[t]{0.4\textwidth}\centering

\vspace{-0.5cm}
\begin{figure}[H]
    \includegraphics[width=\textwidth]{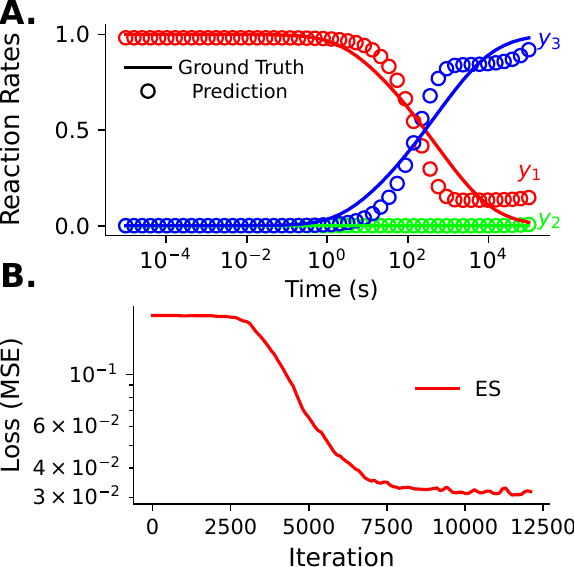}
    \caption{\textbf{Application of ES to the ROBER reaction network.} \textbf{A}  Chemical reaction rates, $y_1, y_2, y_3$, are simulated for exponentially long timespans. The dotted lines represent the ES predictions. \textbf{B} Illustrates the loss over training iterations with ES.}
    \vspace{-1.5cm}
    \label{fig:rober}
\end{figure}

\end{minipage}
\end{wrapfigure}

We next applied ES to a more common stiff problem, the Robertson problem (ROBER) \citep{robertson1967}. ROBER models a stiff system of three nonlinear ordinary differential equations describing the kinetics of an autocatalytic reaction network. ROBER is used extensively as a prototypical example of a stiff ODE since it is simple to describe but hard to solve, requiring implicit methods and adaptive timestepping \citep{robertson1967}. ROBER is explicitly modelled by a system of three variables, $y_1, y_2, y_3$, with dynamics:
\begin{align*}
    \frac{\text{d}}{\text{dt}} \begin{pmatrix} y_1 \\ y_2 \\ y_3 \end{pmatrix} = \begin{pmatrix} -k_1 y_1 + k_3 y_2 y_3 \\ k_1 y_1 - k_2 y_2^2 - k_3 y_2 y_3 \\ k_2 y_2^2 \end{pmatrix}.
\end{align*}

The constants $k_1, k_2, k_3$ in the specific ROBER test are set to $0.04, 3 \cdot 10^7, 10^4$, respectively. The large difference in scales between these constants causes the dynamics to be highly stiff. Furthermore, the solutions to the ROBER problem are evaluated over a timescale with both varying precision and length, ranging from time $10^{-4}$ to $10^4$ (see Figure \ref{fig:rober}), requiring the ODE solver to use highly varied step sizes.

The goal of ROBER is to reproduce the three trajectories observed with the dynamics above over the exponential timescale using a continuous time neural network. Recently, it has been noted that BP directly applied to ROBER can be unstable and result in poor learning \citep{kim2021stiff}. 
%As in this work, it has been noted that on stiff problems such as ROBER, direct BP can be unstable and not train. 
Approaches to stabilize BP on this problem by accounting for the different scales of the problem dynamics have been proposed and have been very successful \citep{kim2021stiff, yingbo2021comparison}.

Here, we present an alternative, which is the ES methodology applied to the ROBER problem. 
%It is found that ES can train to higher accuracy than BP with more stabilized gradient norm. 
As in \citep{kim2021stiff}, we train a network with six hidden layers with 5 hidden neurons. Learned results are summarized in Figure~\ref{fig:rober}. Note that while final prediction accuracy for ES have room for improvement, they exceed those that are found with BP through neural ODEs \citep{kim2021stiff}. These results demonstrate that ES may be of value for solving stiff neural ODE problems. Stability of ES is likely due to the fact that it does not require temporal propagation of gradients. In particular, a direct application of BP to the ROBER problem requires computing adjoint and gradient dynamics backwards in time, which introduces high stiffness and instability.
% Reasons for this instability are more well investigated in \citep{kim2021stiff, yingbo2021comparison}. 

\section{Discussion}

\subsection{Future Work}
A natural next step is to investigate the use of ES in training biological neural networks to solve a broader variety of tasks, for example other working memory, cognition, or visual perception tasks. If ES succeeds in enabling efficient training of such networks in these tasks, it could open the door to new optimization-based methods for dealing with how components of biological circuits can combine to subserve computations.  
%Experimenting with training larger neural networks and for longer timeframes would also test the bounds of backpropagation even further than in this work.
%Unrelated to new choices of task, we also envision experiments related to forming biophysical neuron circuitry. 
For example, one could include multiple neuron types (e.g., neurons with different spiking and bursting properties) and study how these neurons may play different roles in the network's overall computation. This could build on interesting work with, for example, excitatory/inhibitory cells with different timescales \citep{miller2020generalized, wang2002}. \EcommentDone{there is a paper i think from ken miller lab at Columbia that trains networks with separate e and i cell types we could also cite ...}

% Our method is reliant on gradient computation through Monte-Carlo estimation. Although such estimation is proven to scale better with dimensionality of the network over non-sampling approaches such as finite difference, it can still have high variance in practice. Pathwise estimates, in particular, often have a lower variance than score-estimators \citep{mohamed2020monte}. In this work, our aim was to set the foundations for smoothed gradient estimation as a general approach for training stiff and complex neural networks. Thus, we did not apply any significant variance reduction approaches the estimate \ref{derivative_formula}. We demonstrated that ES with this naive estimate could compare to and in certain settings outperform BP in some cases (see Results and Appendix \ref{detailssims}). Throughout our work, we simply used a normal distribution with a single hyperparameter defining the standard deviation for sampling. Using other distributions could yield better results. In particular, we believe using an adaptive approach where the distribution ``learns'' how to optimally sample while training could reduce variance or better account for differences in stiffness in the parameters.

% \textcolor{blue}{(James) Mention (1) variance reduction, and other extension of the method; (2) limitation on samples (but mention why requiring many samples make sense biologically, and how ES still scales better than FD);... }. 

Future directions also include expanding the application of ES to more abstract questions in neural network learning.  Beyond steps toward BNNs, we have additionally demonstrated the versatility of ES by applying it successfully to both stiff and non-stiff neural ODE problems, which indicates that application to many more dynamical systems is possible. Moreover, considering the recent advancements in theoretical tools for studying deep learning generalization~\cite{jiang2020neurips,jacot2018neural,advani2020high,petzka2021relative,pezeshki2021gradient,baratin2021implicit,tsuzuku2020normalized,dziugaite2017computing,zhou2020towards,xie2020diffusion,liu2022beyond,ghosh2023gradient}, an intriguing future direction would be to employ these tools to explore the generalization properties of ES. Inspired by recent studies that illustrate how noise and perturbations~\cite{yu2022robust,orvieto2022anticorrelated,wu2020adversarial,dapello2021neural,zhou2019toward,liu2022beyond} and particularly how stochastic gradient descent operates on implicitly convolved loss~\cite{kleinberg2018alternative,haruki2019gradient}  enhances generalization, the generalization of ES can be considered.

\EcommentDone{Here is some material from above, that it seemed to me best to work into future work / discussion material ... above ... want to take a crack?}

Our work suggests that future research should either aim to provide methods to facilitate training with BP in the cases mentioned (networks with abrupt ODE dynamics), or focus on improving evolutionary algorithms as an alternative. We believe the latter avenue may be promising in the future as dropping the BP requirement gives much freedom, allowing researchers to explore pathology (explosive) neural models or complex loss functions that may not be sufficiently well behaved for BP. Another major advantage to removing the need for BP is that it allows the network to be treated as a black box, so that efficient or standard packages such as \textit{Neuron} and \textit{Brian} could be used for simulation \citep{hines2001neuron, stimberg2019brian2} without the need to implement explicit differentiation or other techniques to compute gradients directly.  

% \textbf{CONSIDER FORCE LEARNING}

%\subsection{Broader Impact}
% Broad impact is not needed if you are not submitting to NeurIPS

%While our work primarily contributes to the fundamental understanding of how biological neural networks function, through the development of new methodologies to train biophysically accurate models, we do not anticipate immediate societal or ethical implications. Consequently, immediate measures to safeguard our model are not necessary at this juncture. Nevertheless, the potential long-term ramifications, particularly for associated fields such as neuroscience and machine learning, are noteworthy. In the realm of neuroscience, the ES method might offer new avenues to understanding neural computation, which could lead to significant strides in unraveling the mysteries of the brain. This, in turn, could guide us towards finding treatments for various neurological conditions. However, just as such advancements could prove beneficial, there is also the potential for misuse if left unregulated. From a machine learning perspective, the ES method proposes an alternative approach to facilitate the training of highly stiff and nonlinear models. This could imply that more capable entities might leverage such efficient learning methods to tackle even more complex tasks, which could raise concerns if not properly managed. Finally, like any data-driven tool, our method is subject to the biases inherent in its training data. Thus, careful consideration and scrutiny must be given to the datasets used, to ensure fairness and avoid perpetuating existing biases.

\subsection{Summary and Conclusion}

In this study, our primary aim is to investigate the applicability of evolutionary approaches for training biological neural networks (BNNs), contrasting them to backpropagation (through the neural ODEs approach and BPTT). Enabling efficient training of these models could be valuable to experimental and computational neuroscientists since it could provide new insights into the emergent role different neurons can play in networks and how different neuron models can be combined together to form circuitry well suited to solve specific problems. We find that BNNs have a number of properties that make direct BP difficult to use: spiking and voltage oscillations can be very abrupt, causing gradient exploding, long timeframes can make memory consumption too high, loss landscapes can be noisy, requiring smoothing or many stochastic trials to resolve well, and, finally, the underlying ODEs are hard to simulate backwards in time, even for very short time intervals (around 10 ms in the present setting). We show that ES, by reformulating the supervised learning task as minimization of a ``smoothed loss'' in the parameter space (see Methods), can successfully compute gradients through a forward-pass only Monte-Carlo estimation, circumventing stability issues often encountered by BP. We provide evidence of the efficacy of ES across several scenarios, including those where BP is inapplicable. Specifically, we demonstrate that ES can effectively train a single Morris--Lecar neuron (Figure~\ref{fig:landscape-train}), despite the jagged loss landscape and unstable gradient that impede BP. Furthermore, we showed that ES can match or outperform BP in training Morris--Lecar neurons to solve a discrete integrator task, a well-known task in neuroscience (Figure~\ref{fig:integrator_task}). These two use cases could be used in the future to evaluate other methodologies for training networks of biophysical neurons. For example, in this work they demonstrate sources of of issues with BP, including exploding gradients and irreversibility of dynamics.

% \textcolor{blue}{This could be attributed to the unstable adjoint inherent in the system (Figure~\ref{fig:nde_adjoint}).} 
Lastly, we extended the applicability of ES beyond BNNs to stiff or non-stiff neural ODE problems, with comparable and in some settings improved results relative to BP (Figures~\ref{fig:nde_acc} and~\ref{fig:rober}). 
%We found that when compared to direct numerical estimates of the gradient such as finite difference, ES requires substantially less samples (Figure~\ref{fig:nde_acc}). 
This builds on a growing body of literature that suggests the applicability of ES for overparameterized or problems with difficult ``loss landscapes'' as a viable alternative to BP.
Taken together, our results underscore the potential of ES as a versatile tool for training networks that exhibit high stiffness, noisiness, non-differentiability, or irreversibility for short timeframes, where traditional BP faces significant challenges. 

\subsection{Acknowledgements} 
We acknowledge the following sources of funding and support: the NIH BRAIN initiative grant 1RF1DA055669 (JH,ESB), National Science Foundation grant CRCNS IIS-2113003 (JH,ES), National Science Foundation grant EFMA-2223495 (JH,ES), and the Departments of Applied Mathematics and Electrical and Computer Engineering at the University of Washington. James Hazelden is an NSF Graduate Fellow

We thank Kyle Aitken, Stefan Mihalas and Alexander Hsu for their insights and comments on this research.

\subsubsection*{Code Availability}

Our code is available upon request.
% All code for our simulations is provided at: \url{https://github.com/meeree/smoothGrad}.

% \bibliographystyle{unsrt}
\bibliography{ref_main}

\section{Appendix}

\subsection{Single Neuron Problem}

\begin{figure}[t]
    \centering
    \includegraphics[width=0.9\textwidth]{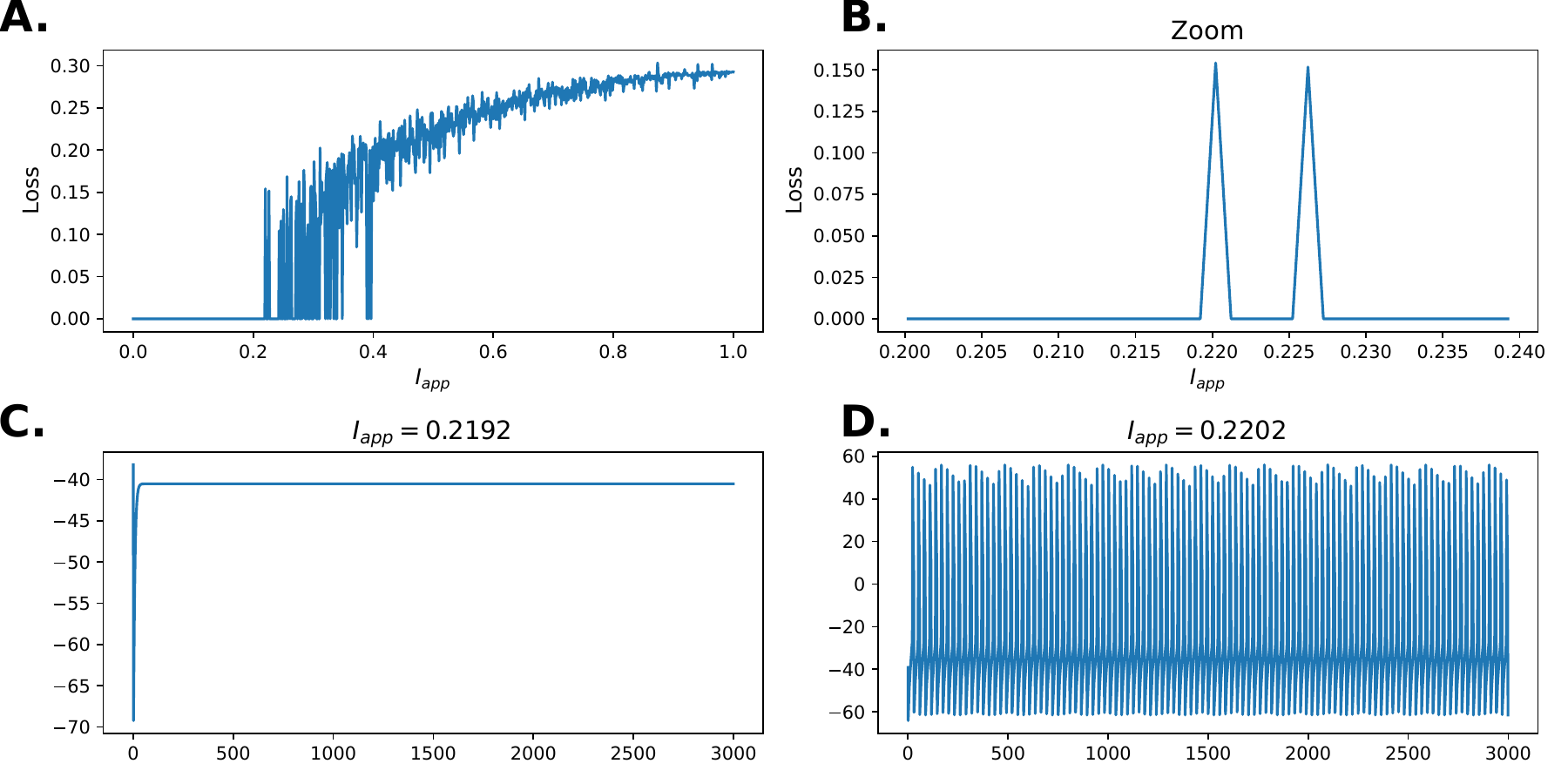}
    \caption{Explanation of random spikes in loss landscape for ``turn off'' single neuron task. \textbf{A} demonstrates the same loss landscape as in Figure \ref{fig:landscape-train}.A. \textbf{B} shows a zoom in where ``spikes'' in the loss can be seen. Under further analysis of the dynamics in the case with low loss and with high loss \textbf{C} and \textbf{D} respectively, we see that in the former case the neuron is not spiking, while in the latter it is constantly spiking for the 3000 ms timeframe. Such behavior is explainable by the bifurcation dynamics of this neuron model where they transition discontinuously from no spiking to constant spiking. Note that one way to alleviate this would be to run many randomized trials, which is a natural feature of evolutionary approaches.} 
    \label{fig:suppzoom}
\end{figure}

\label{sec:methodsbnns}

The following Morris--Lecar neuron physiological parameters are used:
    
\begin{table}[H]
    \centering

    \begin{tabular}{c|c}
         \textbf{Parameter} & \textbf{Value} \\
         $\phi$ & $0.04$ \\ 
         $V_1$ & -1.2 mV \\
         $V_2$ & 18 mV \\ 
         $V_3$ & 2 mV \\
         $V_4$ & 30 mV \\
         $C$ & 20
    \end{tabular}
\end{table}

For explosive parameters, we modify $\phi$ by increasing it ($0.4$ was used, i.e. $10$ times faster dynamics of the gating variable $w$). We can decrease $V_2$ or $V_4$ to get a similar net effect of faster neuronal dynamics. 

\textbf{Numerical Method:}

We use a trapezoidal/leapfrog scheme to efficiently solve the ODE with a fixed step size that can be quite large (~0.1 ms). The scheme has been previously used for similarly stiff problems and Hodgkin--Huxley-like neurons \cite{stinchcombe2017multiplexing}. 

In particular, the model can be written in the simple form 
\begin{align*}
C \frac{{dV}}{{dt}} &= G V - E, \\
\frac{{dw}}{{dt}} &= \phi \cdot \frac{{w_{\infty}(V) - w}}{{\tau_w(V)}},
\end{align*}
where 
\begin{align*}
    G &= g_L + g_{Ca} m_\infty + g_K w, \\
    E &= g_L V_L + g_{Ca} V_{Ca} m_\infty + g_K V_K w.
\end{align*}
The method works by assuming $w$ is updated on the ``half timestep'' and $V$ on the ``full timestep.'' I.e., $V$ is updates at the times $0, \Delta t, 2 \Delta t, ...$ and $w$ is updated at the times $\Delta t/2, 3 \Delta t/2, 5 \Delta t/2, ...$ and both are treated as constant on the others' respective update time. We formulate the ODE update using an implicit equations, but, using this half-stepping scheme, it can made into an explicit update rule.

Let $V_1 = V(t + \Delta t), V_0 = V(t)$. We treat $w$ as constant during the update, as mentioned, since the gating variable updates on the half timestep. Then, if we apply the Trapezoidal rule and the fundamental theorem to approximately integrate both sides of the $V$ update above, we get 
\[C (V_1 - V_0) = \frac{\Delta t}{2} G (V_1 + V_0) - E.\]
Hence, solving for $V_2$, the value at the next timestep of the voltage, 
\[V_1 = \frac{V_0 (\frac{\Delta t}{2}G + C) - E}{C - \frac{\Delta t}{2} G}. \]
Thus, we can use an explicit formula to update but have the added stability of an implicit method, which in practice allows for much large timesteps \cite{stinchcombe2017multiplexing}. 

We likewise do the same to update $w$ on the half timestep. Let $w_1 = w(t + \Delta t + \Delta t /2), w_0 = w(t + \Delta t /2)$. Then, integrating as above, 
\[w_1 - w_0 = \phi \cdot \frac{w_\infty - \frac{\Delta t}{2} (w_1 + w_0)}{\tau_w}.\]
Hence, 
\[w_1 = \frac{\phi w_\infty + w_0 (1 - \phi \frac{\Delta t}{2})}{\tau_w (1 + \frac{\Delta t}{2} \phi)}\]

\subsection{Spiral Neural ODE}

\label{detailssims}

\begin{figure}[H]
    \centering
    \includegraphics[width=0.5\textwidth]{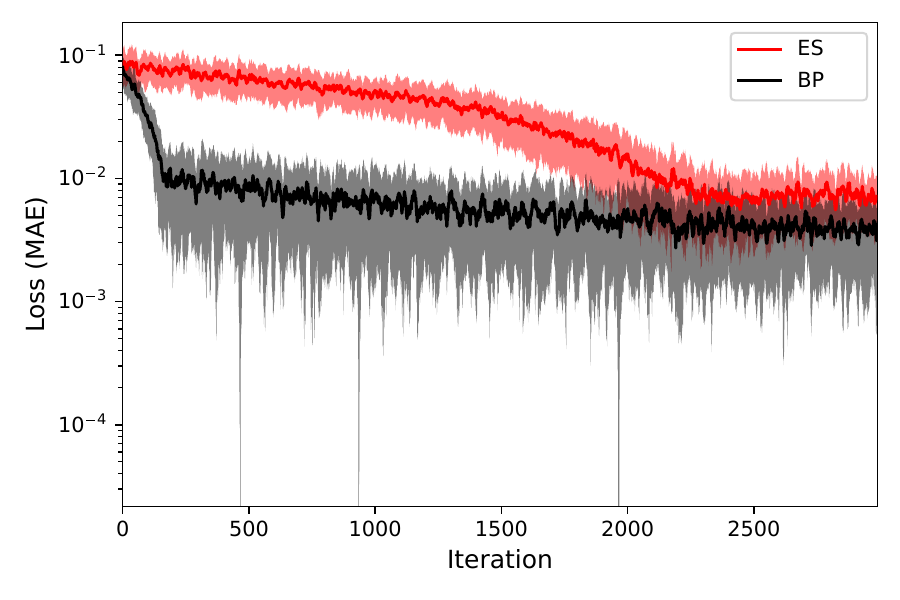}
    \caption{Training comparison for NDEs problem with 50 hidden neurons. 5 training runs were performed with randomized initializations of the network and random batches. The standard deviation $\sigma$ for ES sampling was set to 0.1.}
    \label{fig:nde_acc_50}
\end{figure}

Minimal hyperparameter sweeping was performed on the learning rate when using BP to obtain the fastest convergence. Once this learning rate was obtained (1e-3), it was used idenitcally for each simulation using ES. This method was choosing because it constrains ES to perform well under the same condition as BP, without additionally sweeping the learning rate of ES in each case.

\begin{figure}[H]
    \centering
    \includegraphics[width=0.5\textwidth]{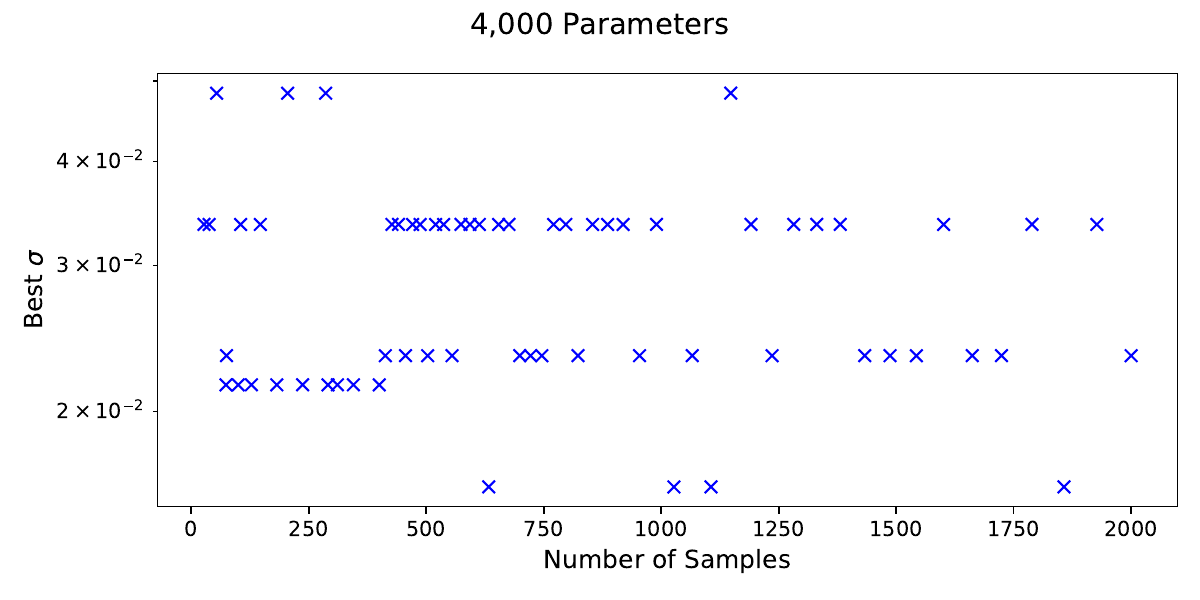}
    \caption{Choice of standard deviation for sampling with ES for each value of $S$ in Figure \ref{fig:nde_acc}.B-C.}
    \label{fig:std_choice}
\end{figure}

For each simulation result in \ref{fig:nde_acc}, ES was used with the normal distribution $p_\theta = \mathcal{N}(\theta; \sigma)$. The standard deviation $\sigma$ is a scalar and was chosen as a hyperparameter that was individually swept in each case, with the fixed learning rate above. The choices of $\sigma$ in each case with fastest convergence to the goal loss (0.04 MAE) in Figure \ref{fig:nde_acc}.B-C are shown in Figure \ref{fig:std_choice} above. We note that there is not a clear trend, but all variances have around the same scale ($10^{-2}$) even though exponential sweeping values from $10^{-1}$ to $10^{-4}$ were considered. To solve the system, the dopri5 method was used in torchdiffeq with \textit{tol} and \textit{atol} set to 1e-4.

\end{document}